\documentclass[12pt,draftcls,onecolumn]{IEEEtran}

\usepackage{graphicx}
\usepackage{array}
\usepackage{mathrsfs}
\usepackage{amsfonts}
\usepackage{setspace}
\usepackage{amsmath}
\usepackage{amssymb}
\usepackage{mathrsfs}
\input xy
\xyoption{all}

\newtheorem{theorem}{Theorem}
\newtheorem{lemma}{Lemma}

\newtheorem{definition}{Definition}
\newtheorem{corollary}{Corollary}

\newtheorem{remark}{Remark}
\newtheorem{example}{Example}
\newtheorem{problem}{Problem}

\begin{document}

 \title{Technical Report: NUS-ACT-11-001-Ver.2:\\
 Fault-tolerant Cooperative Tasking for Multi-agent Systems\\Preprint, Submitted for publication}
 %\tnotetext[label1]{}
 \author{Mohammad~Karimadini,
        and~Hai~Lin,
\thanks{M. Karimadini and H. Lin are both from the Department of Electrical and Computer Engineering, National University of Singapore,
Singapore. Corresponding author, H. Lin {\tt\small elelh@nus.edu.sg}
} } \maketitle \thispagestyle{empty} \pagestyle{empty}

\begin{abstract}
A natural way for cooperative tasking in multi-agent systems is
through a top-down design by decomposing a global task into subtasks
for each individual agent such that the accomplishments of these
subtasks will guarantee the achievement of the global task. In our
previous works \cite{Automatica2010-2-agents-decomposability,
TAC2011-n-agents-decomposability}, we presented necessary and
sufficient conditions on the decomposability of a global task
automaton between cooperative agents. As a follow-up work, this
paper deals with the robustness issues of the proposed top-down
design approach with respect to event failures in the multi-agent
systems. The main concern under event failure is whether a
previously decomposable task can still be achieved collectively by
the agents, and if not, we would like to investigate that under what
conditions the global task could be robustly accomplished. This is
actually the fault-tolerance issue of the top-down design, and the
results provide designers with hints on which events are fragile
with respect to failures, and whether redundancies are needed. The
main objective of this paper is to identify necessary and sufficient
conditions on failed events under which a decomposable global task
can still be achieved successfully. For such a purpose, a notion
called passivity is introduced to characterize the type of event
failures. The passivity is found to reflect the redundancy of
communication links over shared events, based on which necessary and
sufficient conditions for the reliability of cooperative tasking
under event failures are derived, followed by illustrative examples
and remarks for the derived conditions.
\end{abstract}

%\begin{keyword}
%Global task decomposition \sep Fault-tolerant \sep Decentralized
%bisimilarity control \sep Multi-agent system.
%%% keywords here, in the form: keyword \sep keyword
%
%%% MSC codes here, in the form: \MSC code \sep code
%%% or \MSC[2008] code \sep code (2000 is the default)
%
%\end{keyword}

%%%%%%%%%%%%%%%%%%%%%%%%%%%%%%%%%%%%%%%%%%%%%%%%%%%%%%%%%%%%%%%%%%%%%%
%%%%%%%%%%%%%%%%%%%%%%%%%%%%%%%%%%%%%%%%%%%%%%%%%%%%%%%%%%%%%%%%%%%%%%
%%%%%%%%%%%%%%%%%%%%%%%%%%%%%%%%%%%%%%%%%%%%%%%%%%%%%%%%%%%%%%%%%%%%%%
\section{Introduction}\label{Introduction}
Multi-agent system has emerged as a hot research area with strong
support from a wide range of applications such as power grids,
transportation networks, ubiquitous computation, and multi-robot
systems \cite{Lima2005, Georgios2009, Ji2009}. The significance of
multi-agent systems roots in the power of parallelism and
cooperation between simple components that lead to sophisticated
capabilities and more robustness and functionalities than individual
multi-skilled agents \cite{Lesser1999, Choi2009, Kazeroon2009}.
%However, the design of multi agent systems has imposed new
%theoretical and practical challenges that fall beyond the
%traditional path planning, output regulation, or formation control
%\cite{Tabuada2006, Belta2007}.
%
%So far, most design approaches for multi-agent systems follow the
%bottom-up paradigm and mainly draw inspirations from the swarming
%behaviors of natural systems \cite{Bonabeau1999, Reynolds1987,
%Hinchey2005, Jadbabaie2003}. It is quite intuitive, but sometimes
%unexpected behavior could emerge, for which one needs to redesign
%the local interactions and re-investigate the group behavior by
%usually empirical or sometimes analytical studies. Since it is
%difficult to guarantee a desired global behavior or avoid
%undesirable behaviors from the bottom-up design, this trial and
%error process may quickly make the design intractable. On the other
%hand, a natural way for cooperative tasking in multi-agent systems
%is through a top-down design, by decomposing a global task
One of the key problems in multi-agent systems is top-down
cooperative tasking, through which a global task is decomposed
into
subtasks for each individual agent such that the accomplishments of
these subtasks will guarantee the achievement of the global task.

For such a purpose, in our previous work
\cite{Automatica2010-2-agents-decomposability} a top-down design
approach for multi-agent cooperative tasking was proposed for two
agents and then generalized in
\cite{TAC2011-n-agents-decomposability} into an arbitrary finite
number of agents.
As the main contribution,
\cite{Automatica2010-2-agents-decomposability,
TAC2011-n-agents-decomposability} identified necessary and
sufficient conditions under which a deterministic task automaton is
decomposable with respect to parallel composition and natural
projections into local event sets, namely, the task automaton is
bisimilar to the parallel composition of its natural projections.
Moreover, it has been shown that if the task automaton is
decomposable and local supervisors are designed to satisfy local
specification automata, then the entire closed loop system satisfies
the original global specification. It is worth noting here that the
determinism of global task automaton does not reduce its
decomposability in the sense of bisimulation into its
decomposability in the sense of language equivalence
\cite{Mukund2002}, nor into separability of its language
\cite{Willner1991}, since in general local task automata obtained by
natural projection could be nondeterministic, in general (see
Example \ref{Undecomposable DC3-truncation} in the Appendix).

% The reason is that even for a deterministic
%global automaton, the parallel composition of its local projections
%might be nondeterministic. Another point is that the local task
%automata, obtained by natural projection of the global deterministic
%automaton could be nondeterministic and are not necessarily
%simulated by the global automaton. Nevertheless, the objective is
%bisimulation of the global task automaton and parallel composition
%of its local task automata, in order to guarantee the global
%specification by local controllers.

%Once the supervisor is synthesized, the controlled system is
%required to safely work in spite of some faults and changes in the
%system.
Once a multi-agent system is designed, its safety becomes a crucial
property across the agents in order to prevent the uncompensable
consequences for the system and users. Failures on the other hand
are usually unavoidable due to the large scale nature and complex
interactions among the distributed agents. It is therefore very
important to introduce some degree of redundancy into the design so
as to achieve fault-tolerance. Towards this end, this paper
represents a continuation of the works in
\cite{Automatica2010-2-agents-decomposability,
TAC2011-n-agents-decomposability}, and deals with the robustness
issues of the proposed top-down design approach with respect to
event failures in the multi-agent systems. The main concern under
failure is whether a previously decomposable task still can be
achieved collectively by the agents. Please note that no global
information on failures is assumed, and each agent is only aware of
failures around itself and just trying to accomplish its previously
assigned subtask (assume that the global task is decomposable before
failures, and subtasks are obtained, accordingly). An interesting
question is whether these agents can achieve the original global
task in spite of event failures. If not, we would like to ask under
what conditions the global task could be robustly accomplished. This
is actually the fault-tolerance issue of the top-down design, and
the results provide designers hints on which events are fragile with
respect to failures, and whether redundancies are needed for sharing
of some events. It is desired to share as few number of events as
possible through the communication links to reduce the bandwidth,
and hence, the cost of the design. The main objective of this paper
is to identify necessary and sufficient conditions on failed events
under which a decomposable global task can still be achieved
successfully between cooperative agents.

This work differs from diagnosability and isolation problems
\cite{Sampath1995} whose interest is on detection and identification
of the type of faults. In this work the faults are known and the
question is the tolerance of systems in spite of the faults. It also
differs from reliable supervisory control \cite{Takai2000, Liu2010}
that seeks the minimal number of
 supervisors required for correct functionality of the supervised
 systems. Another different problem is robust supervisory control
  \cite{Lin1993} that considers the plant as a set of possible plants and
 designs supervisor applicable for the whole range of plants.

This work is related to the fault-tolerant supervisory control that
has been widely studied in the context of discrete event systems.
For examples, \cite{Darabi2003} proposed switching to another
supervisor after fault detection. In another work,
\cite{Rohloff2005}, the author proposed to re-synthesis the
supervisor upon the fault occurrence. A framework for fault-tolerant
supervisory control has been proposed in \cite{Lafortune1990} and
further explored in \cite{Jensen2003} by enforcing given
specifications for non-faulty and faulty parts of the plant to
ensure that the plant recovers from any fault within a bounded
delay, such that the recovered plant is equivalent to the non-faulty
plant. In \cite{Qin2008} a fault is modeled as an uncontrollable
event, that its occurrence causes a faulty behavior. They provided a
necessary and sufficient condition for the existence of supervisor
under failures, based on controllability, observability and
relative-closure, together with the notions of state-stability
\cite{Brave1990, Ozveren1991}, and language-stability
\cite{Kumar1993, Willner1995}. In \cite{Saboori2005}, a fault
recovery result has been proposed by introducing normal, transient
and recovery modes, such that the language of the closed loop
systems is equal to a given language of the normal mode. Most of
these works however address the
 language specifications and deal with decentralized supervisory
 control with distributed supervisor and monolithic plant.

In this paper, continuing the works in
\cite{Automatica2010-2-agents-decomposability,
TAC2011-n-agents-decomposability}, it is firstly observed that a
necessary condition for preserving the decomposability is that the
failed events  can only be shared events and could be those that are
only received from the other agents or sent to others, redundantly.
In other words, a necessary condition for failed events is that they
are not produced by the sensors/actuators of the corresponding
agent, and that the failed events are not sent to other agents,
unless there exist some alternative agents to relay them. We will
call these events as passive events in the agent. Passive events
indeed refer to the shared events through redundant communication
links. Based on this notation, it seems that the failure of passive
events have no effect on decomposability, as they do not fail in the
sender agents and the receiver is just  no longer informed about
those events. However, it will be shown that although passivity of
failure events is a necessary condition for preserving the
decomposability, some additional conditions are required for the
task automaton to remain decomposable. The intuitive reason is that
when a shared event fails, the corresponding agent can no longer use
its information as a part of decision making on the order or switch
between transitions. Moreover, the failure should satisfy some
criteria to ensure that after the failures, the parallel composition
of local automata neither generates a new string that is not allowed
in the global automaton, nor prevents a string that is allowed in
the global task automaton. In particular, while the passivity of
failed events is a necessary condition to preserve the
decomposability, it is shown that for a deterministic task automaton
that experiences failures on passive events, the task automaton
remains decomposable if and only if any required decisions on
order/switch between any pair of events can be accomplished by at
least one of the agents after failure; no illegal string is allowed
and no legal string is prevented by the composition of local task
automata, after the failure. This work generalizes the preliminary
work on task decomposition under failure in
\cite{CCC2011-n-agents-Event-Failure}, from two agents into an
arbitrary finite number of agents, providing the proofs and
illustrative examples. It furthermore shows that under the passivity
of failed events together with the proposed conditions, a previously
achieved task automaton can be still achieved by the team of agents.

The rest of the paper is organized as follows. Section
\ref{Preliminaries}
 provides preliminary lemmas, notations,
definitions and recalls the necessary and sufficient conditions on
decomposition of an automaton with respect to parallel composition
and local event sets. The fault-tolerant task decomposability and
multi-tasking problems are formulated in Section \ref{Problem
formulation}. Sections \ref{Task Decomposability under event
failure} presents the main result on decomposability under event
failures and introduces the necessary and sufficient conditions
under which a decomposable task automaton remains decomposable in
spite of event failures, followed by illustrative examples for each
condition. Next, it is shown in Sections \ref{Cooperative tasking
under event failure} that under passivity and the proposed
conditions, if a previously decomposable task automaton has been
achieved globally by local controllers, it will remain satisfied, in
spite of event failures. As a special case, Section \ref{Special
case - 2-agent} provides more insight on global decision making on
selections and orders of transitions, in a two-agent case. Finally,
the paper concludes with remarks and discussions in Section
\ref{Conclusion}. The proofs of lemmas are given in the Appendix.

%\section{REVIEW OF TASK AUTOMATON DECOMPOSITION}\label{Preliminary}
\section{Preliminaries}\label{Preliminaries}
The proposed top-down approach in
\cite{Automatica2010-2-agents-decomposability,
TAC2011-n-agents-decomposability} investigated the deterministic
global task automata and introduced necessary and sufficient
conditions under which the task automaton is decomposable with
respect to parallel composition and natural projections into local
event sets, such that the parallel composition of local task
automata bisimulates the global task automaton. It was also shown
that fulfilment of local task automata, leads to satisfaction (in
the sense of bisimulation) of the global task automaton. We then
first recall the definition of an automaton \cite{Kumar1999}.
\begin{definition}(Automaton)
A deterministic automaton is a tuple $A := \left(Q, q_0, E, \delta
\right)$ consisting of a set of states $Q$; an initial state $q_0\in
Q$; a set of events $E$ that causes transitions between the states,
and a transition relation $\delta \subseteq Q \times E\times Q$,
with partial map $\delta: Q \times E \to Q$, such that $(q, e,
q^{\prime})\in \delta$ if and only if state $q$ is transited to
state $q^{\prime}$ by event $e$, denoted by
$q\overset{e}{\underset{}\rightarrow } q^{\prime}$ (or $\delta(q, e)
= q^{\prime}$). A nondeterministic automaton is a tuple $A :=
\left(Q, q_0, E, \delta \right)$ with a partial transition map
$\delta: Q \times E \to 2^Q$, and if hidden transitions
($\varepsilon$-moves) are also possible, then a nondeterministic
automaton with hidden moves is defined as $A := \left(Q, q_0,
E\cup\{\varepsilon\}, \delta \right)$ with a partial map $\delta: Q
\times (E\cup\{\varepsilon\} )\to 2^Q$. For a nondeterministic
automaton the initial state can be generally from a set
$Q_0\subseteq Q$. Given a nondeterministic automaton $A$, with
hidden moves, the $\varepsilon$-closure of $q\in Q$, denoted by
$\varepsilon^*_A(q)\subseteq Q$, is recursively defined as: $q\in
\varepsilon^*_A(q)$; $q^{\prime}\in \varepsilon^*_A(q)\Rightarrow
\delta(q^{\prime}, \varepsilon)\subseteq \varepsilon^*_A(q)$. The
transition relation can be extended to a finite string of events,
$s\in E^*$, where $E^*$ stands for $Kleene-Closure$ of $E$ (the set
of all finite strings over elements of $E$). For an automaton
without hidden moves, $\varepsilon^*_A(q) = \{q\}$, and the
transition on string is inductively defined as
$\delta(q,\varepsilon) = q$ (empty move or silent transition), and
$\delta(q,se)=\delta(\delta(q,s),e)$ for $s\in E^*$ and $e\in E$.
For an automaton $A$, with hidden moves, the extension of transition
relation on string, denoted by $\delta: Q \times E^*\to 2^Q$, is
inductively defined as: $\forall q \in Q, s \in E^*, e\in E$:
$\delta(q, \varepsilon):= \varepsilon^*_A(q)$ and $\delta(q,
se)=\varepsilon^*_A(\delta(\delta(q,s),e)) = \left[ \overset{} {
\underset{q^{\prime}\in \delta(q, s)}{\cup} }\left\{
 \overset{}  { \underset{ q^{\prime\prime}\in \delta(q^{\prime}, e) } {\cup}
 }  \varepsilon^*_A(q^{\prime\prime}) \right\}\right]$.

The operator $Ac(.)$ \cite{Cassandras2008} is then defined by
excluding the states and their attached transitions that are not
reachable from the initial state as $Ac(A) = \left(Q_{ac}, q_0, E,
\delta_{ac} \right)$ with $Q_{ac}=\{q\in Q|\exists s\in E^*, q \in
\delta (q_0, s)\}$ and $\delta_{ac}=\delta|Q_{ac}\times E\rightarrow
Q_{ac}$, restricting $\delta$ to the smaller domain of $Q_{ac}$.
Since $Ac(.)$ has  no effect on the behavior of the automaton, from
now on we take $A = Ac(A)$.
\end{definition}

We focus on deterministic global task automata that are simpler to
be characterized, and cover a wide class of specifications. The
qualitative behavior of a deterministic system is described by the
set of all possible sequences of events starting from the initial
state. Each such a sequence is called a string, and the collection
of strings represents the language generated by the automaton,
denoted by $L(A)$. The existence of a transition over a string $s\in
E^*$ from a state $q\in Q$ is denoted by $\delta(q, s)!$.
Considering a language $L$, by $\delta(q, L)!$ we mean that $\forall
\omega \in L: \delta(q, \omega)!$.

%Next, the successive event pair and adjacent event pair for an
%automaton are defined as follows.
%
%\begin{definition}(Successive and adjacent event pairs)
%Two events $e_1$ and $e_2$ are called successive events if $\exists
%q\in Q: \delta(q,e_1)! \wedge \delta(\delta(q,e_1),e_2)!$ or
%$\delta(q,e_2)! \wedge \delta(\delta(q,e_2),e_1)!$. They are called
%adjacent events if $\exists q\in Q:\delta(q,e_1)! \wedge
%\delta(q,e_2)!$.
%\end{definition}

To compare the task automaton and its decomposed automata, we use
the bisimulation relations \cite{Cassandras2008}.
\begin{definition}(Simulation and Bisimulation)
\label{simulation} Consider two automata $A_i=( Q_i, q_i^0$, $E,
\delta _i)$, $i=1, 2$. A relation $R\subseteq Q_1 \times Q_2$ is
said to be a simulation relation from $A_1$ to $A_2$ if $(q_1^0,
q_2^0) \in R$, and $\forall\left( {q_1 ,q_2 } \right) \in R,
\delta_1(q_1, e)= q'_1$, then $\exists q_2^{\prime}\in Q_2$ such
that $\delta_2(q_2, e)=q'_2, \left( {q'_1 ,q'_2 } \right) \in R$.

If $R$ is defined for all states and all events in $A_1$, then $A_1$
is said to be similar to $A_2$ (or $A_2$ simulates $A_1$), denoted
by $A_1\prec A_2$ \cite{Cassandras2008}.

If $A_1\prec A_2$, $A_2\prec A_1$, with a symmetric relation, then
$A_1$ and $A_2$ are said to be bisimilar (bisimulate each other),
denoted by $A_1\cong A_2$ \cite{Zhou2006}. In general, bisimilarity
implies languages equivalence but the converse does not necessarily
hold \cite{Alur2000}.
\end{definition}

In these works natural projection is used to obtain local tasks,
since each agent has limited degree of sensing and actuation and
hence it is provided with local information and functionalities:
those events inside its local event set. Each agent may share some
events with its neighbors to facilitate the cooperative control,
using interactions between the connected agents. Natural projection
is defined formally as follows.
\begin{definition}(Natural Projection on String)
Consider a global event set $E$ and its local event sets $E_i$,
$i=1,2,...,n$, with $E=\overset{n}{\underset{i=1}{\cup} } E_i$.
Then, the natural projection $p_i:E^*\rightarrow E_i^*$ is
inductively defined as
 $p_i(\varepsilon)=\varepsilon$, and
 $\forall s\in E^*, e\in E:p_i(se)=\left\{
  \begin{array}{ll}
  p_i(s)e & \hbox{if $e\in E_i$;} \\
  p_i(s) & \hbox{otherwise.}
  \end{array}
\right.$

Accordingly, inverse natural projection $p_i^{-1}: E_i^* \to
2^{E^*}$ is defined on an string $t\in E_i^*$ as $p_i^{-1}(t):=
\{s\in E^*|p_i(s) = t\}$.
\end{definition}

The natural projection is also defined on automata as $P_i:
A\rightarrow A$, where, $A$ is the set of finite automata and
$P_i(A_S)$ are obtained from $A_S$ by replacing its events that
belong to $E\backslash E_i$ by $\varepsilon$-moves, and then,
merging the $\varepsilon$-related states, forming equivalent classes
defined as follows.
\begin{definition}(Equivalent class of states, \cite{Morin1998})\label{Equivalent class of states}
Consider an automaton $A_S=(Q, q_0, E, \delta)$ and a local event
set $E_i\subseteq E$. Then, the relation $\sim_{E_i}$ is the
equivalence relation on the set $Q$ of states such that $\delta(q,
e) = q^{\prime}\wedge e\notin E_i\Rightarrow q\sim_{E_i}
q^{\prime}$, and $[q]_{E_i}$ denotes the equivalence class of $q$
defined on $\sim_{E_i}$. In this case, $q$ and $q^{\prime}$ are said
to be $\varepsilon$-related. The set of equivalent classes of states
over $\sim_{E_i}$, is denoted by $Q_{/\sim_{E_i}}$ and defined as
$Q_{/\sim_{E_i}} = \{[q]_{E_i}|q\in Q\}$.
%When it is clear from the context,
%$\sim_i$ and $[q]_i$ are used instead of $\sim_{E_i}$ and
%$[q]_{E_i}$, respectively.
\end{definition}

The natural projection is then formally defined on an automaton as
follows.
\begin{definition}(Natural Projection on Automaton)\label{Natural Projection on Automaton}
Consider an automaton $A_S=(Q, q_0, E, \delta)$ and a local event
set $E_i\subseteq E$. Then, $P_i(A_S)=(Q_i = Q_{/\sim_{E_i}},
[q_0]_{E_i}, E_i, \delta_i)$, with $\delta_i([q]_{E_i}, e) =
[q^{\prime}]_{E_i}$ if there exist states $q_1$ and $q_1^{\prime}$
such that $q_1\sim_{E_i} q$, $q_1^{\prime}\sim_{E_i} q^{\prime}$,
and $\delta(q_1, e) = q^{\prime}_1$.
\end{definition}

To investigate the interactions of transitions between automata,
particularly between $P_i(A_S)$, $i = 1, \ldots, n$, the
synchronized product of languages is defined as follows.
%
%\begin{definition}(Path Automaton)
%A sequence $q_1\overset{e_1}{\underset{}\rightarrow }
%q_2\overset{e_2}{\underset{}\rightarrow
%}...\overset{e_n}{\underset{}\rightarrow }q_n$ is called a path
%automaton and is defined as $PA(q_1, s)= (\{q_1,...,q_n\}, \{q_1\}$,
%$\{e_1,...,e_n\}$,$ \delta_{PA})$ with  $\delta_{PA}(q_i,
%e_i)=q_{i+1}$, $i=1,...,n-1$.
%\end{definition}

\begin{definition}(Synchronized product of languages
\cite{Willner1991}) Consider a global event set $E$ and local event
sets $E_i$, $ i = 1, \ldots, n$, such that $E =
\overset{n}{\underset{i=1}{\cup} } E_i$. For a finite set of
languages $\{L_i \subset E_i^*\}_{i = 1}^n$, the synchronized
product (language product) of $\{L_i\}$, denoted by
$\overset{n}{\underset{i=1}{|} } L_i$, is defined as
$\overset{n}{\underset{i=1}{|} } L_i = \{s \in E^*|\forall
i\in\{1,\ldots, n\}: p_i(s)\in L_i\} =
\overset{n}{\underset{i=1}{\cap} } p_i^{-1}(L_i)$.
\end{definition}

%Using the language product, it is then possible to characterize the
%language of parallel composition of two automata in terms of their
%languages as
%\begin{lemma}(\cite{Willner1991}) \label{langauge of parallel automata}
%Consider two automata $A_1$ and $A_2$, with respective event sets
%$E_1$ and $E_2$. Then $L(A_1||A_2) = L(A_1)|L(A_2) =
%P_1^{-1}(L(A_1))\cap P_2^{-1}(L(A_2))$ with $P_i: E_1\cup E_2 \to
%E_i$, $i = 1, 2$.
%\end{lemma}
%
%This characterization is particularly true for the special case,
%when the composed automata have no branches, and consisted of only
%sequences of states and events.
%\begin{definition}(Path Automaton)
%A sequence $q_1\overset{e_1}{\underset{}\rightarrow }
%q_2\overset{e_2}{\underset{}\rightarrow
%}...\overset{e_n}{\underset{}\rightarrow }q_n$ is called a path
%automaton and is defined as $PA(q_1, s)= (\{q_1,...,q_n\}, \{q_1\}$,
%$\{e_1,...,e_n\}$,$ \delta_{PA})$ with  $\delta_{PA}(q_i,
%e_i)=q_{i+1}$, $i=1,...,n-1$.
%\end{definition}
%The language of parallel composition of two path automata $PA(q_1,
%s)$ and $PA^{\prime}(q_1^{\prime}, s^{\prime})$ is then defined as
%$L(PA(q_1, s)||PA^{\prime}(q_1^{\prime}, s^{\prime})) =
%s|s^{\prime}$.

Then, capturing the interactions of agents, parallel composition
(synchronized product) is used for two purposes: first to define the
decomposition (as the parallel composition of local tasks should be
equivalent to the original task), and second, to define the top-down
cooperative control, such that the parallel composition local closed
loop systems be equivalent to the global specification.

The parallel composition (synchronous product) is a way of modeling
of interactions between agents as it allows local agents to transit
on their own private events and restricts them to synchronize on the
shared events, those events that are required for cooperation on
common actions or decision makings on orders or selections between
events.
%Parallel composition can be used for both language model
%(pertaining the language equivalence-based cooperative control) as
%well as automaton model (utilizing the bisimulation-based top-down
%design).

\begin{definition} (Parallel Composition) \label{parallel
composition}\\
Let $A_i=\left( Q_i,q_i^0,E_i,\delta _i\right)$, $i=1,2$ be
automata. The parallel composition (synchronous composition) of
$A_1$ and $A_2$ is the automaton $A_1||A_2=\left( Q = Q_1 \times
Q_2, q_0 = (q_1^0, q_2^0), E = E_1 \cup E_2, \delta\right)$, with
$\delta$ defined as $\forall (q_1, q_2)\in Q$, $e\in E$:
$\delta(\left(q_1, q_2), e\right)=\left\{
\begin{array}{ll}
    \left(\delta_1(q_1, e), \delta_2(q_2, e)\right), & \hbox{if $\left\{\begin{array}{ll}
    \delta _1(q_1,e)!, \delta _2(q_2,e)! \\
     e\in E_1 \cap E_2
\end{array}\right.$};\\
    \left(\delta_1(q_1, e), q_2\right), & \hbox{if $\delta _1(q_1,e)!, e\in E_1 \backslash E_2$;} \\
    \left(q_1, \delta_2(q_2, e)\right), & \hbox{if $\delta _2(q_2,e)!, e\in E_2 \backslash E_1$;} \\
    \hbox{undefined}, & \hbox{otherwise.}
\end{array}\right.$

The parallel composition of $A_i$, $i=1,2,...,n$ is called parallel
distributed system (or concurrent system), and is defined based on
the associativity property of parallel composition
\cite{Cassandras2008} as $\overset{n}{\underset{i=1}{\parallel\ }
}A_i=A_1\parallel\ ...\parallel\   A_n = A_n\parallel \left(A_{n-1}
\parallel \left( \cdots \parallel \left( A_2\parallel
A_1 \right)\right)\right)$.

The set of labels of local event sets containing an event $e$ is
called the set of locations of $e$, denoted by $loc(e)$ and is
defined as $loc(e) =\{i\in\{1,\ldots, n\}| e\in E_i\}$.
%In the state space of
%$\overset{n}{\underset{i=1}{\parallel\ } }A_i$, for any state $q\in
%\overset{n}{\underset{i=1}{\prod\ } }Q_i$, $q[i]$ means the $i-th$
%component of $q$.
\end{definition}

In this  sense, the decomposability of an automaton with respect to
parallel composition and natural projections is defined as follows.
\begin{definition}(Automaton decomposability)\label{Automaton decomposability}
 A task automaton $A_S$ with the event set $E$ and local event sets
 $E_i$, $i=1,..., n$, $E = \overset{n}{\underset{i=1}{\cup} } E_i$, is
said to be decomposable with respect to parallel composition and
natural projections if $\overset{n}{\underset{i=1}{\parallel} } P_i
\left( A_S \right) \cong A_S$.
\end{definition}
%Preliminary notations and definitions for automaton, simulation and
%bisimulation, natural projection on strings and automata can be
%found in \cite{Automatica2010-2-agents-decomposability} and
%references therein, and are not repeated here, due to limitation in
%space.

In general, the task automaton can be nondeterministic.
Decomposition of nondeterministic automaton however is very
difficult to be characterized, due to interleaving of
nondeterministic transitions between local task automata. For
deterministic case, necessary and sufficient conditions for the
decomposability of a deterministic task automaton $A_S$ were
proposed in \cite{Automatica2010-2-agents-decomposability} with
respect to two cooperative agents and then generalized into an
arbitrary finite number of agents, as follows.

\begin{lemma}(Corollary $1$ in
\cite{TAC2011-n-agents-decomposability}): \label{Task
Decomposability Theorem for n agents} A deterministic automaton $A_S
= \left( {Q, q_0 , E = \bigcup\limits_{i = 1}^n {E_i , \delta } }
\right)$ is decomposable with respect to parallel composition and
natural projections $P_i$, $i=1,...,n$ such that $A_S \cong \mathop
{||}\limits_{i = 1}^n P_i \left( {A_S } \right)$ if and only if
$A_S$ satisfies the following decomposability conditions ($DC$):
\begin{itemize}
\item $DC1$: $\forall e_1,
e_2 \in E, q\in Q$: $[\delta(q,e_1)!\wedge \delta(q,e_2)!]\\
\Rightarrow [\exists E_i\in\{E_1, \cdots, E_n\}, \{e_1,
e_2\}\subseteq E_i]\vee[\delta(q, e_1e_2)! \wedge \delta(q,
e_2e_1)!]$;
\item $DC2$: $\forall e_1, e_2 \in E,  q\in Q$, $s\in E^*$: $[\delta(q,
e_1e_2s)!\vee \delta(q, e_2e_1s)!]\\ \Rightarrow [\exists
E_i\in\{E_1, \cdots, E_n\}, \{e_1, e_2\}\subseteq E_i]\vee [
\delta(q, e_1e_2s)!\wedge \delta(q, e_2e_1s)!]$;
\item $DC3$: $\delta(q_0, \mathop |\limits_{i = 1}^n p_i \left( {s_i }
\right))!$, $\forall \{s_1, \cdots, s_n\}\in \tilde L\left( {A_S }
\right)$, $\exists s_i, s_j \in \{s_1, \cdots, s_n\}, s_i \neq s_j$,
where, $\tilde L\left( {A_S } \right) \subseteq L\left( {A_S }
\right)$ is the largest subset of $L\left( {A_S } \right)$ such that
$\forall s\in \tilde L\left( {A_S } \right), \exists s^{\prime} \in
\tilde L\left( {A_S } \right),\;\exists E_i ,E_j  \in \left\{ {E_1
,...,E_n } \right\},
 i \ne j,p_{E_i  \cap E_j } \left( s \right)$ and $p_{E_i  \cap E_j } \left( s^{\prime} \right)$  start with the same
 event, and
\item $DC4$: $\forall i\in\{1,...,n\}$, $x, x_1, x_2 \in Q_i$, $x_1\neq x_2$,
$e\in E_i$, $t\in E_i^*$, $\delta_i (x, e)=  x_1$, $\delta_i (x, e)=
x_2$: $\delta_i (x_1, t)! \Leftrightarrow \delta_i(x_2, t)!$.
 \end{itemize}
\end{lemma}

Intuitively, the decomposability condition $DC1$ means that for any
decision on selection between two transitions there should exist at
least one agent that is capable of the decision making, or the
decision should not be important (both permutations in any order be
legal). $DC2$ says that for any decision on the order of two
successive events before any string, either there should exist at
least one agent capable of such decision making, or the decision
should not be important, i.e., any order would be legal for
occurrence of that string. The condition $DC3$ means that the
interleaving of strings from local task automata that synchronize on
the same first appearing shared event, should not allow a string
that is not allowed in the original task automaton. In other words,
$DC3$ is to ensure that an illegal behavior (an string that does not
appear in $A_S$) is not allowed by the team (does not appear in
$\mathop {||}\limits_{i = 1}^n P_i \left( {A_S } \right)$). In this
condition, $\mathop |\limits_{i = 1}^n p_i (s_i )$ is a language and
stands for the interleaving or language product \cite{Willner1991}
of strings $p_i (s_i )$, defined as $\mathop |\limits_{i = 1}^n p_i
(s_i ) = \overset{n}{\underset{i=1}{\cap} } p_i^{-1}(p_i (s_i ))$.
The last condition, $DC4$, deals with the possible nondeterminisms
in $P_i \left( {A_S } \right)$ (Please note that here, $A_S$ is
considered to be deterministic, while $P_i \left( {A_S } \right)$
can be nondeterministic, as it will be explained in Example
\ref{decomposable automaton with nondeterministic projection}).
$DC4$ ensures the determinism of bisimulation quotient of local task
automata, in order to guarantee the symmetry of simulation relations
between $A_S$ and $\mathop {||}\limits_{i = 1}^n P_i \left( {A_S }
\right)$. By providing this symmetry property, $DC4$ guarantees that
a legal behavior (an string in $A_S$) is not disabled by the team
(appears in $\mathop {||}\limits_{i = 1}^n P_i \left( {A_S }
\right)$).

In \cite{TAC2011-n-agents-decomposability} it was also shown that
for a decomposable task automaton, if local controllers exist such
that each local closed loop system (parallel composition of local
plant and local controller automata) satisfies its local task
(bisimulates the corresponding local task automaton), then the
controlled team of the agents will satisfy the global specification,
as it is stated in the following lemma.

\begin{lemma}\label{Top-down-bisimilarity}(Theorem $2$ in
\cite{TAC2011-n-agents-decomposability}): Consider a plant,
represented by a parallel distributed system
$\overset{n}{\underset{i=1}{\parallel} }A_{P_i}$, with given local
event sets $E_i$, $i=1,...,n$, and let the global specification is
given by a deterministic task automaton $A_S$, with
$E=\overset{n}{\underset{i=1}{\cup} }E_i$. If $DC1$-$DC4$ are
satisfied, then designing local controllers $A_{C_i}$, so that
$A_{C_i}\parallel A_{P_i}\cong P_i(A_S)$, $i=1,\cdots, n$, derives
the global closed loop system to satisfy the global specification
$A_S$, i.e., $\overset{n}{\underset{i=1}{\parallel}
}(A_{C_i}\parallel A_{P_i})\cong A_S$.
\end{lemma}

\begin{remark}\label{Remarks for determinism}
It is known that bisimulation implies language equivalence and that
bisimulation of deterministic automata is reduced to their language
equivalence. Now, one question is that whether for a deterministic
task automaton its decomposability in the sense of bisimulation
 (stated in Lemma \ref{Task Decomposability Theorem for n agents})
is reduced to its decomposability in the sense of language
equivalence ($L(A_S) = L(\overset{n}{\underset{i=1}{\parallel} } P_i
\left( A_S \right))$) or its language separability ($L(A_S)
  = \overset{n}{\underset{i=1}{|} } L(P_i \left( A_S \right))$).
Furthermore, it is interesting to know whether the proposed top-down
cooperative control, in Lemma \ref{Top-down-bisimilarity}, is
reduced into a top-down approach in the sense of language
equivalence. As it is illustrated in the Appendix, although in
general, decomposability in the sense of bisimulation implies the
decomposability in the sense of language equivalence, the reverse is
not always true, in spite of determinism of automaton. For the
top-down cooperative control, on the other hand, under the proposed
decomposability conditions, the bisimulation-based approach is
reduced to the language equivalence one, as the deterministic task
automaton can be represented by its langauge.

To elaborate these remarks, we first highlight that the natural
projection may impose emerging properties that do not exist in the
original automaton. For example, local task automata may have some
new strings that do not appear in the original automaton, i.e.,
$A_S$ does not necessarily simulates $P_i \left( A_S \right)$.
Moreover, local task automata may become nondeterministic, even if
the original task automaton is deterministic. The decomposability of
$A_S$, however, concerns with bisimilarity of $A_S$ and
$\overset{n}{\underset{i=1}{\parallel} } P_i \left( A_S \right)$,
that may hold even if $P_i(A_S)\not\prec A_S$, or the local task
automata are nondeterministic for some agents, as it is shown
through examples in the Appendix.
\end{remark}

\section{Problem formulation}\label{Problem formulation}
In the previous section we recalled the conditions for task
automaton decomposability to be used in top-down cooperative
control. A natural follow-up question is that if after such
decomposition, some of the events fail in some agents, then whether
the global task automaton will still remain decomposable with
respect to new set of events. And, if not, what are the conditions
for preserving the decomposability. In order to address this
problem, we first need to investigate the failure on events. In
general, an event $e$ can be either private ($|loc(e)| = 1$) or
shared ($|loc(e)|> 1$). Failure of private events fails the
decomposability as it causes the failure in the whole team of
agents. Failure on a shared event, on the other hand, may or may not
lead to a global failure, depending on whether the failed event is
redundant or not. When an event is a sensor reading; or actuator
command, or it is sent to other agents with no other alternative
links, then the failure on this event stops its global evolutions.
In the following, we will introduce a class of failures that are
investigated in this paper.
\begin{definition}(Event failure)
Consider an automaton $A = (Q, q_0, E, \delta)$. An event $e\in E$
is said to be failed in $A$ (or $E$), if $F(A) = P_{\Sigma}(A) =
P_{E\backslash e}(A) = (Q, q_0, \Sigma = E\backslash e, \delta^F)$,
where, $\Sigma$, $\delta^F$ and $F(A)$ denote the post-failure event
set, post-failure transition relation and post-failure automaton,
respectively. A set $\bar{E}\subseteq E$ of events is then said to
be failed in $A$, when for $\forall e\in \bar{E}$, $e$ is failed in
$A$, i.e., $F(A) = P_{\Sigma}(A_i) = P_{E\backslash \bar{E} }(A) =
(Q, q_0, \Sigma = E\backslash \bar{E}, \delta^F)$.
\end{definition}

Considering a parallel distributed plant $A:=
\overset{n}{\underset{i=1}{||} } A_i = (Z, z_0, E =
\overset{n}{\underset{i=1}{\cup} } E_i, \delta_{||})$ with local
agents $A_i = (Q_i, q_0^i, E_i, \delta_i)$, $i = 1, \ldots, n$.
Failure of $e$ in $E_i$ is said to be passive in $E_i$ (or $A_i$)
with respect to $\overset{n}{\underset{i=1}{||} } A_i$, if $E =
\overset{n}{\underset{i=1}{\cup} } \Sigma_i$. An event whose failure
in $A_i$ is a passive failure is called a passive event in $A_i$.

The notion of passivity, can be interpreted as communication
redundancy as it is stated as follows.
\begin{remark}\label{Meaning of Passivity}
To interpret the passivity more formally, let $snd_e(i)$ and
$rcv_e(i)$ respectively denote the set of labels that $A_i$ sends
$e$ to those agents and the set of labels that $A_i$ receives $e$
from  their agents, defined as $snd_e(i)=\{j\in \{1,...,n\}|A_i
\hbox{ sends } e \hbox{ to } A_j\}$ and $rcv_e(i)=\{j\in
\{1,...,n\}|i\in snd_e(j)\}$. Then, an event $e$ is passive in $A_i$
if $rec_e(i)\neq \emptyset$ (i.e., the $i-th$ agent does not receive
$e$ from its own sensor/actuator readings, but from another agent),
and $\forall k \in snd_e ( i )$: $\exists j \in \{1, \cdots ,n
\}\backslash \{ i, k \},k \in snd_e ( j )$ (i.e., if the $i-th$
agent is a relay for transmission of $e$, for any receiver agent,
there exist another agent to send $e$). In this set-up a passive
failure excludes the failed event $e$ from the corresponding local
event set $E_i$ while it  makes its respective transitions hidden in
$F(A_i)$. Therefore, from definition of parallel composition, the
transitions on other agents can contribute to form the global
transitions in $\overset{n}{\underset{i=1}{||} } F(A_i)$, since only
in this way there will be no synchronization constraint on the rest
of agents in $\overset{n}{\underset{i=1}{||} }F(A_i)$.
\end{remark}

 Moreover, the definition of passivity implies that
 the passivity of failed events is a necessary condition for
 evolution of global transitions after failures,
as it is stated in the following lemma.
\begin{lemma}(Global transitions after local failures) \label{Passivity properties}
Consider a parallel distributed plant $A:=
\overset{n}{\underset{i=1}{||} } A_i = (Z, z_0, E =
\overset{n}{\underset{i=1}{\cup} } E_i, \delta_{||})$ with local
agents $A_i = (Q_i, q_0^i, E_i, \delta_i)$, $i = 1, \ldots, n$. If
no global transitions in $\overset{n}{\underset{i=1}{||} } A_i$ are
disabled in $\overset{n}{\underset{i=1}{||} } F(A_i)$ (i.e.,
$\forall z_1, z_2 \in Z$, $\forall e\in E$, $\delta_{||}(z_1, e) =
z_2$, then $\delta_{||}^F(z_1, e) = z_2$), then all event failures
are passive, i.e., the passivity of local event failures are
necessary for preserving the global transitions.
\end{lemma}
\begin{proof}
See the proof in the Appendix.
\end{proof}

The problem of decomposability under event failures is now defined
as follows.
\begin{problem}(Decomposability under event failures)
\label{Decomposability under event failure} Let a deterministic task
automaton $A_S  = (Q, q_0, E = \mathop{\cup}\limits_{i = 1}^n E_i,
\delta)$ is decomposable with respect to parallel composition and
natural projections $P_i$, $i=1, \ldots, n$. Then, does the global
task automaton $A_S$ remain decomposable in spite of failure of
events $\{a_{i,r}\}$, $r\in\{1,...,n_i\}$ in local event sets $E_i$,
$i\in\{1, \ldots, n\}$? i.e., if $A_S \cong
\overset{n}{\underset{i=1}{\parallel} }P_i (A_S)$, then does $A_S
\cong \overset{n}{\underset{i=1}{\parallel} }F(P_i(A_S))$ always
hold true?, and if not, what are the conditions for such
decomposability?
\end{problem}

The next interesting question is the cooperative control under event
failure, defined as \begin{problem}(Cooperative tasking under event
failure) \label{Cooperative tasking under event failure-problem}
Consider a concurrent plant $A_P:= \mathop {||}\limits_{i = 1}^n
A_{P_i}$ and a decomposable deterministic task automaton $A_S = (Q,
q_0, E = \mathop {\cup}\limits_{i = 1}^n E_i, \delta)\cong \mathop
{||}\limits_{i = 1}^n P_i (A_S )$, and suppose that local controller
automata $A_{C_i}$, $i = 1, \ldots, n$ exist such that each local
closed loop system satisfies its corresponding local task, i.e.,
$A_{C_i}||A_{P_i}\cong P_i(A_S)$, $i = 1, \ldots, n$. Assume
furthermore that $\bar{E}_i = \{a_{i,r}\}$ fail in $E_i$,
$r\in\{1,...,n_i\}$. Then, does the team still can fulfill the
global task, in spite of failures, without redesigning the
controller automata, i.e., $\mathop {||}\limits_{i = 1}^n
F(A_{P_i}||A_{C_i})\cong A_S$?, and if not, what are the conditions
to preserve the satisfaction of the global specification?
\end{problem}

 These problems will be addressed in
the following two sections.

\section{Task Decomposability under event failures}\label{Task Decomposability under event failure}
According to definition of passivity, for any local event set $E_i$,
$\Sigma_i$ excludes any passive failed events $e$ from $E_i$ , while
%keeps those events that are non-passive in it, and moreover,
 the effect of this failure on
$P_i(A_S)$ is defined as the projection of $A_S$ into $E_i\backslash
e$ (instead of $E_i$), leading to $P_{E_i\backslash e}(A_S)$.
 %This
%means that the transitions on passive failed events in $P_i(A_S)$
%are captured by replacing the corresponding transitions in
%$F(P_i(A_S))$ with $\varepsilon$ and merging the
%$\varepsilon-related$ states.
%Whereas, if the failed event $e$ is non-passive in $E_i$, then its
%corresponding transitions stop in $F(P_i(A_S))$.

In this set up evolution of global transitions in
$\mathop{||}\limits_{i = 1}^n F(P_i \left( {A_S } \right))$ relies
on the passivity of failed events,
%and disabling global transitions on
%non-passive failed events,
 as it is expected and stated in Lemma
\ref{Passivity properties}. The reason is that due to definition of
parallel composition, evolution of global transitions requires the
failures to be passive, since passive failed events are excluded
from the corresponding local event set and the local task automaton
is projected to the rest of events. For non-passive failed events,
on the other hand, since they are not received from other agents,
and hence are not excluded from the local event set, but their
transitions are stopped, then due to synchronization restriction in
definition of parallel composition, the global transitions cannot
evolve on them.

Consequently, as highlighted in Lemma \ref{Passivity properties},
passivity of failed events is a necessary condition for the task
automaton to remain decomposable after the failure.
%The reason is
%that any failure on any non-passive event disables global
%transitions defined on those events.

Moreover, when all failed events are passive, due to definition of
passivity, Problem \ref{Decomposability under event failure} can be
transformed into the standard decomposition problem to find the
conditions under which $A_S \cong \mathop {||}\limits_{i = 1}^n
P_{E_i\backslash \bar{E}_i} (A_S )$. Accordingly, the conditions on
the global task automaton to preserve the decomposability under
event failures, are reduced into their respective decomposability
conditions in Lemma \ref{Task Decomposability Theorem for n agents},
as the following lemmas.

\begin{lemma}\label{EF1,2 and DC1,2}
Consider a deterministic task automaton $A_S = (Q, q_0, E = \mathop
{\cup}\limits_{i = 1}^n E_i, \delta)$. Assume that $A_S$ is
decomposable, i.e., $A_S \cong \mathop {||}\limits_{i = 1}^n P_i
(A_S )$, and suppose that $\bar{E}_i = \{a_{i,r}\}$ fail in $E_i$,
$r\in\{1,...,n_i\}$,  and $\bar{E}_i$ are passive for $i\in\{1,
\ldots, n\}$. Then, following two expressions are equivalent:
\begin{enumerate}
\item
\begin{itemize}
\item $EF1$: $\forall e_1,
e_2 \in E, q\in Q$: $[\delta(q, e_1)!\wedge \delta(q, e_2)!]\\
\Rightarrow [\exists E_i\in\{E_1, \ldots, E_n\}, \{e_1,
e_2\}\subseteq E_i\backslash \bar{E}_i]\vee[\delta(q, e_1e_2)!
\wedge \delta(q, e_2e_1)!]$;
\item $EF2$: $\forall e_1, e_2 \in E,  q\in Q$, $s\in E^*$: $[\delta(q,
e_1e_2s)!\vee \delta(q, e_2e_1s)!]\\ \Rightarrow [\exists
E_i\in\{E_1, \ldots, E_n\}, \{e_1, e_2\}\subseteq E_i\backslash
\bar{E}_i]\vee [ \delta(q, e_1e_2s)!\wedge \delta(q, e_2e_1s)!]$.
\end{itemize}
\item \begin{itemize}
\item $DC1_{\Sigma}$: $\forall e_1,
e_2 \in E, q\in Q$: $[\delta(q, e_1)!\wedge \delta(q, e_2)!]\\
\Rightarrow [\exists \Sigma_i\in\{\Sigma_1, \ldots, \Sigma_n\},
\{e_1, e_2\}\subseteq \Sigma_i]\vee[\delta(q, e_1e_2)! \wedge
\delta(q, e_2e_1)!]$;
\item $DC2_{\Sigma}$: $\forall e_1, e_2 \in E,  q\in Q$, $s\in E^*$: $[\delta(q,
e_1e_2s)!\vee \delta(q, e_2e_1s)!]\\ \Rightarrow [\exists
\Sigma_i\in\{\Sigma_1, \ldots, \Sigma_n\}, \{e_1, e_2\}\subseteq
\Sigma_i]\vee [\delta(q, e_1e_2s)!\wedge \delta(q, e_2e_1s)!]$.
\end{itemize}
\end{enumerate}
\end{lemma}
\begin{proof}
See the proof in the Appendix.
\end{proof}

Lemma $\ref{EF1,2 and DC1,2}$ gives the simplified versions of $DC1$
and $DC2$ after event failures, with respect to refined local event
sets. Adopting the same $DC3$ for the refined local event sets, it
remains to represent a simplified version of $DC4$ for the local
task automata, after event failures. This condition is stated in the
following lemma.
 \begin{lemma}\label{EF4 and DC4}
Consider a deterministic task automaton $A_S = (Q, q_0, E = \mathop
{\cup}\limits_{i = 1}^n E_i, \delta)$. Assume that $A_S$ is
decomposable, i.e., $A_S \cong \mathop {||}\limits_{i = 1}^n P_i
(A_S )$, and suppose that $\bar{E}_i = \{a_{i,r}\}$ fail in $E_i$,
$r\in\{1,...,n_i\}$,  and $\bar{E}_i$ are passive for $i\in\{1,
\ldots, n\}$. Then, following two expressions are equivalent:
\begin{itemize}
\item $EF4$: $\forall i\in\{1, \ldots, n\}$, $x, x_1, x_2
\in Q_i$, $x_1\neq x_2$, $e\in E_i\backslash \bar{E}_i$, $t_1 \in
\bar{E}_i^*$, $t\in E_i^*$, $\delta_i (x, t_1e)= x_1$, $\delta_i (x,
e)= x_2$: $\delta_i (x_1, t)! \Leftrightarrow \delta_i(x_2, t)!$.
\item $DC4_{\Sigma}$: $\forall i\in\{1, \ldots, n\}$, $x, x_1,
x_2 \in Q_i$, $x_1\neq x_2$, $e\in \Sigma_i$, $t\in \Sigma_i^*$,
$\delta_i^F (x, e)= x_1$, $\delta_i^F (x, e)= x_2$: $\delta_i^F
(x_1, t)! \Leftrightarrow \delta_i^F(x_2, t)!$. Where, $\delta_i^F$
is the transition relation in $F(P_i(A_S))$.
\end{itemize}
\end{lemma}
\begin{proof}
See the proof in the Appendix.
\end{proof}

\begin{remark} $EF4$ is the counterpart of $DC4$ after the event failures,
that handle newly possible nondeterminism in the local task
automata. Any nondeterminism that is propagated from the local task
automata of before the failure, is treated by $DC4$ when $A_S$ is
decomposable.
\end{remark}

Now, combination of Lemmas \ref{Task Decomposability Theorem for n
agents}, \ref{EF1,2 and DC1,2}, and \ref{EF4 and DC4} leads to the
main result on decomposability under event failures as the following
theorem.

\begin{theorem}\label{Decomposability under
event failure-Theorem} Consider a deterministic task automaton $A_S
= (Q, q_0, E = \mathop {\cup}\limits_{i = 1}^n E_i, \delta)$. Assume
that $A_S$ is decomposable, i.e., $A_S \cong \mathop {||}\limits_{i
= 1}^n P_i (A_S )$, and furthermore, assume that $\bar{E}_i =
\{a_{i,r}\}$ fail in $E_i$, $r\in\{1,...,n_i\}$, and $\bar{E}_i$ are
passive for $i\in\{1, \ldots, n\}$. Then, $A_S$ remains
decomposable, in spite of event failures, i.e., $A_S \cong \mathop
{||}\limits_{i = 1}^n F(P_i \left( {A_S } \right))$ if and only if
\begin{itemize}
\item $EF1$: $\forall e_1,
e_2 \in E, q\in Q$: $[\delta(q,e_1)!\wedge \delta(q,e_2)!]\\
\Rightarrow [\exists E_i\in\{E_1, \cdots, E_n\}, \{e_1,
e_2\}\subseteq E_i\backslash \bar{E}_i]\vee[\delta(q, e_1e_2)!
\wedge \delta(q, e_2e_1)!]$;
\item $EF2$: $\forall e_1, e_2 \in E,  q\in Q$, $s\in E^*$: $[\delta(q,
e_1e_2s)!\vee \delta(q, e_2e_1s)!]\\ \Rightarrow [\exists
E_i\in\{E_1, \cdots, E_n\}, \{e_1, e_2\}\subseteq E_i\backslash
\bar{E}_i]\vee [ \delta(q, e_1e_2s)!\wedge \delta(q, e_2e_1s)!]$;
\item $EF3$: $\delta(q_0, \mathop |\limits_{i = 1}^n p_i \left( {s_i }
\right))!$, $\forall \{s_1, \cdots, s_n\}\in \tilde L\left( {A_S }
\right)$, $\exists s_i, s_j \in \{s_1, \cdots, s_n\}, s_i \neq s_j$,
where $\tilde L\left( {A_S } \right) \subseteq L\left( {A_S }
\right)$ is the largest subset of $L\left( {A_S } \right)$ such that
$\forall s\in \tilde L\left( {A_S } \right), \exists s^{\prime} \in
\tilde L\left( {A_S } \right),\;\exists \Sigma_i, \Sigma_j  \in
\left\{ {\Sigma_1 ,...,\Sigma_n } \right\},
 i \ne j,p_{\Sigma_i  \cap \Sigma_j } \left( s \right)$ and
 $p_{\Sigma_i  \cap \Sigma_j } \left( s^{\prime} \right)$  start with the same
 event, and
\item $EF4$: $\forall i\in\{1, \ldots, n\}$, $x, x_1, x_2
\in Q_i$, $x_1\neq x_2$, $e\in E_i\backslash \bar{E}_i$, $t_1 \in
\bar{E}_i^*$, $t\in E_i^*$, $\delta_i (x, t_1e)= x_1$, $\delta_i (x,
e)= x_2$: $\delta_i (x_1, t)! \Leftrightarrow \delta_i(x_2, t)!$.
\end{itemize}
\end{theorem}

\begin{proof}
First, according to Lemma \ref{Passivity properties}, passivity of
$\bar{E}_i$  is a necessary condition for preserving the
decomposability. Now, providing the decomposability of $A_S$ and
passivity of all failed events, due to  definition of passivity, it
leads to $\mathop {||}\limits_{i = 1}^n F(P_i \left( {A_S }
\right))\cong \mathop {||}\limits_{i = 1}^n P_{\Sigma_i} \left( {A_S
} \right) = \mathop {||}\limits_{i = 1}^n P_{E_i\backslash
\bar{E}_i} \left( {A_S } \right)$ that based on Lemmas \ref{Task
Decomposability Theorem for n agents}, \ref{EF1,2 and DC1,2} and
\ref{EF4 and DC4}, it is bisimiar to $A_S$ if and only if $EF1$ -
$EF4$ hold true for the refined local event sets $\{\Sigma_1, \dots,
\Sigma_n\}$.
\end{proof}

\begin{remark}\label{meaning of EF}
$EF1$-$EF4$ are respectively the decomposability conditions
$DC1$-$DC4$, after event failures with respect to parallel
composition and natural projections into refined local event sets
$\Sigma_i = E_i \backslash \bar{E}_i$, $i\in\{1, \ldots, n\}$,
provided the passivity of $\bar{E}_i$, $i\in\{1, \ldots, n\}$.
Condition $EF1$ means that, after failure of some passive events,
for any decision on selection between two transitions there should
exist at least one agent that is capable of the decision making, or
the decision should not be important (both permutations in any order
be legal). $EF2$ says that, after failure of some passive events,
for any decision on the order of two successive events before any
string, either there should exist at least one agent capable of such
decision making, or the decision should not be important, i.e., any
order would be legal for occurrence of that string. The condition
$EF3$ means that, after failure of some passive events, any
interleaving of strings from local task automata that have the same
first appearing shared event, should not allow a string that is not
allowed in the original task automaton. In other words, $EF3$ is to
ensure that, after failure of some passive events, an illegal
behavior (an string that does not appear in $A_S$) is not allowed by
the team (does not appear in $\mathop {||}\limits_{i = 1}^n F(P_i
\left( {A_S }) \right)$). The last condition, $EF4$, ensures the
determinism of bisimulation quotient of local task automaton, in
order to guarantee the symmetry of simulation relations between
$A_S$ and $\mathop {||}\limits_{i = 1}^n F(P_i ( A_S ))$. By
providing this symmetry property, $EF4$ guarantees that, after the
failures, a legal behavior (a string in $A_S$) is not disabled by
the team (appears in $\mathop {||}\limits_{i = 1}^n F(P_i ( A_S )$).
\end{remark}

Following examples illustrate these conditions.
\begin{example}
This example illustrates the notion of passivity and shows a
decomposable automaton that stays decomposable, when an even is
failed passively in one of the local agents and $EF1$-$EF4$ are
satisfied. Consider the automaton $A_S$: \xymatrix@R=0.1cm{
                \ar[r]&\bullet \ar[dr]_{e_2} \ar[r]^{e_1}&\bullet  \ar[r]^{e_2}& \bullet \ar[r]^{a}&\bullet \\
             && \bullet  \ar[r]^{e_1}\ar[dr]_{a} &\bullet \ar[r]^{a}&\bullet \\
             &&& \bullet  \ar[r]^{e_1}&\bullet
             } with
             local event sets $E_1=\{e_1, a\}$ and $E_2=\{e_2, a\}$,
             $E_3=\{a\}$ and communication pattern as
             $\{1, 2\}\in snd_a(3)$, and no other communication links. This automaton is
             decomposable, as the parallel composition of
             $P_1(A_S) \cong \xymatrix@R=0.1cm{
                \ar[r]&\bullet \ar[r]^{e_1}\ar[dr]_{a}&\bullet   \ar[r]^{a}&\bullet \\
             && \bullet  \ar[r]^{e_1}&\bullet }$, $P_2(A_S)\cong \xymatrix@R=0.1cm{
                \ar[r]&\bullet \ar[r]^{e_2} &\bullet \ar[r]^{a}&\bullet}$  and $P_3(A_S)\cong \xymatrix@R=0.1cm{
                \ar[r]&\bullet \ar[r]^{a}&\bullet}$ is $\mathop
{||}\limits_{i = 1}^3 P_i ( A_S )$:\\ \xymatrix@R=0.5cm{
                \ar[r]&\bullet \ar[r]^{e_2}\ar[d]_{e_1}&\bullet   \ar[r]^{a}\ar[d]_{e_1}&\bullet\ar[r]^{e_1} &\bullet\\
             &\bullet \ar[r]^{e_2}&\bullet   \ar[r]^{a}&\bullet
             }
             which is bisimilar to $A_S$.
                Now, assume that $a$ fails in $E_1$. Then
                $EF1$-$EF4$ are satisfied (as $\delta(q, e_2e_1a)! \wedge \delta(q,
e_2ae_1)!$, and hance, $EF1$ and $EF2$ hold true; after the failure,
the interleavings on shared event $a$ impose no illegal strings, and
therefore, $EF3$ is satisfied, and finally $EF4$ is fulfilled since
$F(P_1(A_S))\cong \xymatrix@R=0.1cm{
                \ar[r]&\bullet \ar[r]^{e_1}&\bullet}$,
                $F(P_2(A_S))\cong
\xymatrix@R=0.1cm{
                \ar[r]&\bullet \ar[r]^{e_2} &\bullet
                \ar[r]^{a}&\bullet}$ and $F(P_3(A_S))\\ \cong \xymatrix@R=0.1cm{
                \ar[r]&\bullet \ar[r]^{a}&\bullet}$ are all deterministic), and hence, the parallel composition of
$F(P_1(A_S))$ with $\Sigma_1 = \{e_1\}$, $F(P_2(A_S))$ with
$\Sigma_2 = \{e_2, a\}$, and $F(P_3(A_S))$ with $\Sigma_3 = \{a\}$,
is $\mathop {||}\limits_{i = 1}^3 F(P_i ( A_S ))$:
\xymatrix@R=0.5cm{
                \ar[r]&\bullet \ar[r]^{e_2}\ar[d]_{e_1}&\bullet   \ar[r]^{a}\ar[d]_{e_1}&\bullet\ar[d]_{e_1} \\
             &\bullet \ar[r]^{e_2}&\bullet   \ar[r]^{a}&\bullet
             } that is bisimilar to $A_S$.
      However, if $a$ was failed in $E_3$, then it evolved in none of
      the local task automata and $\mathop
{||}\limits_{i = 1}^3 F(P_i ( A_S ))\ncong A_S$, since $E_3$ is a
source for $a$. Similarly, failure of private events $e_1$ and $e_2$
in $E_1$ and $E_2$, respectively, disables the global transitions on
these events. As another example for non-passive failure, consider
the communication pattern of $1\in snd_a(3)$, $\{2, 3\}\subseteq
snd_a(1)$, while $a$ fails in $E_1$, Then, the parallel composition
of $F(P_1(A_S))$: \xymatrix@R=0.1cm{
                \ar[r]&\bullet \ar[r]^{e_1}&\bullet} with $\Sigma_1 = \{e_1,
                a\}$,
                $F(P_2(A_S))\cong
\xymatrix@R=0.1cm{
                \ar[r]&\bullet \ar[r]^{e_2}&\bullet}$ with $\Sigma_2 = \{e_2\}$, and
$F(P_3(A_S))\cong \xymatrix@R=0.1cm{
                \ar[r]&\bullet \ar[r]^{a}&\bullet}$ with $\Sigma_3 = \{a\}$
                 was
                \xymatrix@R=0.1cm{
                \ar[r]&\bullet \ar[dr]_{e_1} \ar[r]^{e_2}&\bullet \ar[r]^{e_1}&\bullet \\
             && \bullet \ar[ur]_{e_2} }
%\xymatrix@R=0.1cm{
%                &    &     \bullet  \ar[dr]^{e_1}   \\
%\ar[r]&  \bullet \ar[ur]^{e_2} \ar[dr]_{e_1} & &  \bullet   \\
%              && \bullet \ar[ur]_{e_2}&
%                }
              which is not bisimilar to $A_S$. The reason is
                that in this case, in contrast to the fist case, $a$ was not excluded from
                $\Sigma_1$, while $a$ was stopped in $F(P_1(A_S))$. This, due to the synchronization
                constraint in parallel composition, disabled the
                global transitions on $a$.
\end{example}
\begin{example}
This example shows a decomposable automaton that will no longer stay
decomposable after a passive event failure, since $EF1$ is not
satisfied, although other three conditions, $EF2$, $EF3$ and $EF4$,
are fulfilled. Consider the automaton $A_S$: \xymatrix@R=0.1cm{
                \ar[r]&  \bullet \ar[r]^{a} \ar[dr]_{b}&\bullet  \\
             && \bullet                   }
             with local event sets $E_1=\{a\}$, $E_2=\{b\}$, and
            $E_3=\{a, b\}$ with $3\in snd_{a}(1)$ and $3\in
            snd_{b}(2)$, and no other sending and receiving links. This automaton is
             decomposable, as the parallel composition of
             $P_1(A_S)$:\xymatrix@R=0.1cm{ \ar[r]&\bullet
\ar[r]^{a}&\bullet}, $P_2(A_S)$: \xymatrix@R=0.1cm{ \ar[r]&\bullet
\ar[r]^{b}&\bullet} and $P_3(A_S)\cong A_S$ bisimulates $A_S$. Now,
suppose that $a$ is failed in $E_3$. Then, the parallel composition
of $F(P_1(A_S))$: \xymatrix@R=0.1cm{ \ar[r]&\bullet
\ar[r]^{a}&\bullet} with $\Sigma_1 = \{a\}$, $F(P_2(A_S))$:
\xymatrix@R=0.1cm{ \ar[r]&\bullet \ar[r]^{b}&\bullet} with $\Sigma_2
= \{b\}$, and $F(P_3(A_S))$: \xymatrix@R=0.1cm{ \ar[r]&\bullet
\ar[r]^{b}&\bullet} with $\Sigma_3 = \{b\}$, is $\mathop
{||}\limits_{i = 1}^3 F(P_i ( A_S ))$: \xymatrix@R=0.1cm{
                \ar[r]&\bullet \ar[dr]_{a} \ar[r]^{b}&\bullet \ar[r]^{a}&\bullet \\
             && \bullet \ar[ur]_{b} } which is not bisimilar to $A_S$.
The reason is violation of $EF1$, as after the failure of $a$ in
$E_3$, neither there exists an agent that knows both events $a$ and
$b$ to decide on the selection between them, nor both permutations
are legal in $A_S$. If $A_S$ was $A_S$:\xymatrix@R=0.1cm{
                \ar[r]&\bullet \ar[dr]_{b} \ar[r]^{a}&\bullet  \ar[r]^{b}& \bullet  \\
             && \bullet  \ar[r]^{a}&\bullet \\
             }, then, failure of $a$ in $E_3$ had no effect on
             decomposability of $A_S$.
\end{example}

\begin{example}
This example shows a decomposable automaton that will no longer stay
decomposable after a passive failure, as $EF2$ is not satisfied,
although other three conditions, $EF1$, $EF3$ and $EF4$ are
fulfilled. Consider the automaton $A_S$: \xymatrix@R=0.1cm{
                \ar[r]&\bullet \ar[r]^{a}&\bullet\ar[r]^{b}&\bullet}  with
             local event sets $E_1=\{a\}$, $E_2=\{b\}$ and
             $E_3=\{a, b\}$, with $3\in snd_{a}(1)$ and $3\in
            snd_{b}(2)$ with no other sending and receiving links. This automaton is
             decomposable, as the parallel composition of
             $P_1(A_S)$:\xymatrix@R=0.1cm{ \ar[r]&\bullet
\ar[r]^{a}&\bullet}, $P_2(A_S)$: \xymatrix@R=0.1cm{ \ar[r]&\bullet
\ar[r]^{b}&\bullet} and $P_3(A_S)\cong A_S$ bisimulates $A_S$. Now,
suppose that $a$ is failed in $E_3$. Then, the parallel composition
of $F(P_1(A_S))$: \xymatrix@R=0.1cm{ \ar[r]&\bullet
\ar[r]^{a}&\bullet} with $\Sigma_1 = \{a\}$, $F(P_2(A_S))$:
\xymatrix@R=0.1cm{ \ar[r]&\bullet \ar[r]^{b}&\bullet} with $\Sigma_2
= \{b\}$, and $F(P_3(A_S))$: \xymatrix@R=0.1cm{ \ar[r]&\bullet
\ar[r]^{b}&\bullet} with $\Sigma_3 = \{b\}$, is $\mathop
{||}\limits_{i = 1}^3 F(P_i ( A_S ))$:
% \xymatrix@R=0.1cm{
%                &    &     \bullet  \ar[dr]^{a}   \\
%\ar[r]&  \bullet \ar[ur]^{b} \ar[dr]_{a} & &  \bullet   \\
%              && \bullet \ar[ur]_{b}&
%                }
\xymatrix@R=0.1cm{
                \ar[r]&\bullet \ar[dr]_{a} \ar[r]^{b}&\bullet \ar[r]^{a}&\bullet \\
             && \bullet \ar[ur]_{b} }
                 which is not bisimilar to $A_S$.
The reason is violation of $EF2$, as after the failure of $a$ in
$E_3$, neither there exists an agent that knows both events $a$ and
$b$ to decide on the order of them, nor both orders are legal in
$A_S$.
\end{example}

\begin{example}
This example illustrates a decomposable automaton that satisfies
$EF1$, $EF2$ and $EF4$, but it will not remain decomposable after a
passive event failure, due to violation of $EF3$. Consider the
automaton $A_S$:\\ \xymatrix@R=0.1cm{
                \ar[r]&\bullet \ar[dr]_{c} \ar[r]^{a}&\bullet \ar[r]^{b}&\bullet \ar[r]^{c}&\bullet \\
             && \bullet  \ar[r]^{b}&\bullet}
with local event sets $E_1 = \{a, b, c\}$, $E_2 = \{b, c\}$ and $E_3
= \{a, b\}$ and communication pattern $1\in snd_{\{b, c\}}(2)$,
$1\in snd_{a}(3)$, $3\in snd_{b}(2)$, with no other communication
links. $A_S$ is decomposable, as the parallel composition of
$P_1(A_S)\cong A_S$, $P_2(A_S)$: \xymatrix@R=0.1cm{
                \ar[r]&\bullet \ar[dr]_{c} \ar[r]^{b}&\bullet \ar[r]^{c}&\bullet  \\
             && \bullet  \ar[r]^{b}&\bullet}, and $P_3(A_S)$:
\xymatrix@R=0.1cm{
                \ar[r]&\bullet \ar[dr]_{b} \ar[r]^{a}&\bullet \ar[r]^{b}&\bullet  \\
             && \bullet  } is bisimilar to $A_S$.
             Now, assume that $b$ fails in $E_1$. Then, the parallel
             composition of $F(P_1(A_S))$: \xymatrix@R=0.1cm{
                \ar[r]&\bullet \ar[dr]_{c} \ar[r]^{a}&\bullet \ar[r]^{c}&\bullet  \\
             && \bullet } with $\Sigma_1 = \{a,
             c\}$, $F(P_2(A_S))$: \xymatrix@R=0.1cm{
                \ar[r]&\bullet \ar[dr]_{c} \ar[r]^{b}&\bullet \ar[r]^{c}&\bullet  \\
             && \bullet  \ar[r]^{b}&\bullet} with $\Sigma_2 = \{b,
             c\}$ and $F(P_3(A_S))$:\\ \xymatrix@R=0.1cm{
                \ar[r]&\bullet \ar[dr]_{b} \ar[r]^{a}&\bullet \ar[r]^{b}&\bullet  \\
             && \bullet  } with $\Sigma_3 = \{a, b\}$ is $\mathop
{||}\limits_{i = 1}^3 F(P_i ( A_S ))$:
                \xymatrix@C=0.5cm{
&\bullet & \bullet\ar[l]_{b} & \bullet \ar[r]^{b}& \bullet\\
\ar[r]&\bullet  \ar[ur]_{c}\ar[dr]^{b}\ar[r]^{a}& \bullet
\ar[ur]_{c}\ar[dr]^{b}\\
&\bullet & \bullet\ar[l]_{c} & \bullet \ar[r]^{c}& \bullet }\\
 that is no longer bisimilar to $A_S$ due to violation
       of $EF3$ as it contains strings $acb$ and $bc$ that do not appear in $A_S$.
\end{example}

\begin{example}
This example shows a decomposable automaton that does not remain
decomposable against a passive event failure, when it does not
satisfy $EF4$, although it fulfils $EF1$, $EF2$ and $EF3$. Consider
the automaton $A_S$: \xymatrix@R=0.1cm{
                \ar[r]&\bullet \ar[dr]_{c} \ar[r]^{b}&\bullet \ar[r]^{c}&\bullet \\
             && \bullet \ar[r]^{a}&\bullet } with
             local event sets $E_1=\{a, b, c\}$, $E_2=\{a, b\}$ and $E_3=\{b,
             c\}$, with communication structure $1\in snd_{\{a,
             b\}}(2)$, $1\in snd_{c}(3)$, $3\in snd_{b}(2)$, with no
             other communication links. This automaton is
             decomposable, as the parallel composition of
             $P_1(A_S)\cong A_S$, $P_2(A_S)$: \xymatrix@R=0.1cm{
                \ar[r]&\bullet \ar[dr]_{a} \ar[r]^{b}&\bullet \\
             && \bullet} and $P_3(A_S)$: \xymatrix@R=0.1cm{
                \ar[r]&\bullet \ar[dr]_{c} \ar[r]^{b}&\bullet \ar[r]^{c}&\bullet \\
             && \bullet } is bisimilar to $A_S$. Now, assume that $b$ fails  in $E_1$, then the parallel
                composition of $F(P_1(A_S))$: \xymatrix@R=0.1cm{
                \ar[r]&\bullet \ar[r]^{c} \ar[dr]_{c}&\bullet\\
                &&\bullet \ar[r]^{a}&\bullet } with $\Sigma_1=\{a, c\}$, $F(P_2(A_S))\cong  \xymatrix@R=0.1cm{
                \ar[r]&\bullet \ar[dr]_{a} \ar[r]^{b}&\bullet \\
             && \bullet}$ with $\Sigma_2=\{a, b\}$ and $F(P_3(A_S))\cong  \xymatrix@R=0.1cm{
                \ar[r]&\bullet \ar[dr]_{c} \ar[r]^{b}&\bullet \ar[r]^{c}&\bullet \\
             && \bullet }$ with $\Sigma_3=\{b, c\}$ is
                $\mathop
{||}\limits_{i = 1}^3 F(P_i ( A_S ))$:
                \xymatrix@R=0.1cm{
                \ar[r]&\bullet \ar[dl]^{c}\ar[dr]_{c} \ar[r]^{b}&\bullet \ar[r]^{c}&\bullet \\
             \bullet&& \bullet \ar[r]^{a}&\bullet }\\  that is no longer bisimilar to $A_S$ due to violation
       of $EF4$, as there does not exist a deterministic automaton $P_1^{\prime}(A_S)$ such that $P_1^{\prime}(A_S)\cong F(P_1(A_S))$.
\end{example}

\section{Cooperative tasking under event failure}
\label{Cooperative tasking under event failure} So far, we have
presented the necessary and sufficient conditions for a decomposable
task automaton to remain decomposable in spite of passive failures.
%, namely, the fault-tolerance of the divide-part in
%the divide-and-conquer paradigm. Now, following result represents
%the fault-tolerance of the conquer-part in the top-down approach.
Now, assume that the global task automaton is decomposable and local
controllers have been designed in such a way that local
specifications are satisfied, and hence due to Lemma
\ref{Top-down-bisimilarity}, the global specification is satisfied,
by the team. Furthermore, assume that event failures occur on some
shared events, but due to passivity of failed events and
$EF1$-$EF4$, the global task automaton remains decomposable. Then,
the next question is Problem \ref{Cooperative tasking under event
failure-problem} to understand whether, the team is still able to
achieve the global specification. Following result answers this
question.

\begin{theorem}\label{Top-Down Approach Under Failure-n-agents}
Consider a concurrent plant $A_P:= \mathop {||}\limits_{i = 1}^n
A_{P_i}$ and a deterministic task automaton $A_S = (Q, q_0, E =
\mathop {\cup}\limits_{i = 1}^n E_i, \delta)$ as the global
specification. Assume that $A_S$ is decomposable, i.e., $A_S \cong
\mathop {||}\limits_{i = 1}^n P_i (A_S )$, and suppose that local
controller automata $A_{C_i}$, $i = 1, \ldots, n$ have been designed
such that each local closed loop system satisfies its corresponding
local task, i.e., $A_{C_i}||A_{P_i}\cong P_i(A_S)$, $i = 1, \ldots,
n$. Assume furthermore that $\bar{E}_i = \{a_{i,r}\}$ fail in $E_i$,
$r\in\{1,...,n_i\}$, $\bar{E}_i$ are passive for $i\in\{1, \ldots,
n\}$, and $A_S$ satisfies $EF1$-$EF4$. Then, the team can still
achieve its global specification, i.e., $\mathop {||}\limits_{i =
1}^n F(A_{P_i}||A_{C_i})\cong A_S$.
\end{theorem}

\begin{proof}
Firstly, decomposability of $A_S$ and $A_{C_i}||A_{P_i}\cong
P_i(A_S)$, $i = 1, \ldots, n$, due to Lemma
\ref{Top-down-bisimilarity}, implies that $\mathop {||}\limits_{i =
1}^n (A_{P_i}||A_{C_i})\cong A_S$, i.e., the global specification is
satisfied by the team. Moreover, the global specification remains
satisfied, in spite of event failures, if $\bar{E}_i$ are passive
for $i\in\{1, \ldots, n\}$, and $A_S$ satisfies $EF1$-$EF4$, since
$\mathop {||}\limits_{i = 1}^n F(A_{P_i}||A_{C_i})\cong \mathop
{||}\limits_{i = 1}^n P_{E_i\backslash \bar{E_i}}\left(
A_{P_i}||A_{C_i} \right)\cong \mathop {||}\limits_{i = 1}^n
P_{E_i\backslash \bar{E_i}}\left( P_i(A_S) \right)\cong \mathop
{||}\limits_{i = 1}^n F\left( P_i(A_S) \right)\cong \mathop
{||}\limits_{i = 1}^n P_i(A_S) \cong A_S$. In this expression, the
first and the third bisimilarities come from passivity of
$\bar{E}_i$, $i\in\{1, \ldots, n\}$, and the second bisimilarity is
followed from $A_{C_i}||A_{P_i}\cong P_i(A_S)$, $i = 1, \ldots, n$,
definition of natural projection and from the fact $\left(A_1\cong
A_2 \right)\wedge \left(A_3\cong A_4\right)\Rightarrow
\left(A_1\parallel \ A_3\cong A_2\parallel \ A_4\right)$ (Lemma $6$
in \cite{Automatica2010-2-agents-decomposability}). The fourth
equivalence is implied from passivity of $\bar{E}_i$, $i = 1,
\ldots, n$ and $EF1$-$EF4$, and finally, the last bisimilarity is
due to the decomposability assumption of $A_S$.
\end{proof}
\begin{remark}
The significance of Theorem \ref{Top-Down Approach Under
Failure-n-agents} is that under passivity condition and $EF1$-$EF4$,
although local task automata may change after the failure (i.e.,
$F(P_i(A_S))\ncong P_i(A_S)$), the team of agents can satisfy the
global specification, as $\mathop {||}\limits_{i = 1}^n
F(A_{P_i}||A_{C_i})\cong  \mathop {||}\limits_{i = 1}^n
F(P_i(A_S))\cong \mathop {||}\limits_{i = 1}^n P_i(A_S)\cong A_S$.
\end{remark}

\begin{example}
This example illustrates a specification for a team of three agents
that is globally satisfied and remains satisfied in spite of passive
event failures, provided $EF1$-$EF4$. Consider a concurrent plant
$A_P := \mathop {||}\limits_{i = 1}^3 A_{P_i}$ with local plants
$A_{P_1}$: \xymatrix@R=0.1cm{
                \ar[r]&\bullet \ar[r]^{e_1}\ar[dr]_{a}&\bullet   \ar[r]^{a}&\bullet \\
             && \bullet  \ar[r]^{e_1}&\bullet }  with
             $E_1=\{a, e_1\}$, $A_{P_2}$:\\  \xymatrix@R=0.1cm{
             \bullet& \ar[l]_{a} \bullet & \ar[l]_{b} \bullet &\ar[d]&  \bullet
             \ar[r]^{e_2} & \bullet \ar[r]^{a} & \bullet\\
             &&& \bullet \ar[ur]^{b}\ar[ul]_{e_2}\ar[dr]_{a}
             \ar[dl]_{b}\\
              \bullet& \ar[l]_{e_2} \bullet & \ar[l]_{a} \bullet &&  \bullet
             \ar[r]^{e_2} & \bullet \ar[r]^{b} & \bullet
              }
              with
             $E_2=\{a, b, e_2\}$, $A_{P_3}$:\\ \xymatrix@R=0.1cm{
                \ar[r]&\bullet \ar[r]^{e_3}\ar[dr]_{b}&\bullet   \ar[r]^{b}&\bullet \\
             && \bullet  \ar[r]^{e_3}&\bullet }  with
             $E_3=\{b, e_3\}$, having communication pattern $1\in
             send_a(2)$, $3\in
             send_b(2)$, and no more communication links.
Assume that the global specification is given as $A_S$:\\
\xymatrix@R=0.1cm{
             \ar[r]&\bullet \ar[dr]_{a} \ar[r]^{e_1}&\bullet \ar[r]^{a}&\bullet \ar[r]^{e_2}&\bullet \ar[r]^{b}&\bullet \ar[r]^{e_3}&\bullet\\
             && \bullet \ar[dr]_{e_2} \ar[r]^{e_1}&\bullet\ar[r]^{e_2}&\bullet  \ar[r]^{b}&\bullet \ar[r]^{e_3}&\bullet\\
             &&& \bullet \ar[dr]_{b} \ar[r]^{e_1}&\bullet\ar[r]^{b}&\bullet  \ar[r]^{e_3}&\bullet \\
             &&&& \bullet \ar[dr]_{e_3} \ar[r]^{e_1}&\bullet\ar[r]^{e_3}&\bullet\\
             &&&&&\bullet  \ar[r]^{e_1}&\bullet}. $A_S$ is
             decomposable, since the parallel composition of
             $P_1(A_S)\cong \xymatrix@R=0.1cm{
                \ar[r]&\bullet \ar[r]^{e_1}\ar[dr]_{a}&\bullet   \ar[r]^{a}&\bullet \\
             && \bullet  \ar[r]^{e_1}&\bullet }$, $P_2(A_S)\cong\\ \xymatrix@R=0.1cm{
                \ar[r]&\bullet \ar[r]^{a}&\bullet
                \ar[r]^{e_2}&\bullet \ar[r]^{b}&\bullet}$ and $P_3(A_S)\cong \xymatrix@R=0.1cm{
                \ar[r]&\bullet \ar[r]^{b}&\bullet
                \ar[r]^{e_3}&\bullet}$
                is bisimilar to $A_S$. Now,
                taking local controller as $A_{C_i} := P_i(A_S)$, $i = 1, 2, 3$ results in $A_{P_i}||A_{C_i} \cong P_i(A_S)$, $i = 1, 2, 3$
                and $\mathop {||}\limits_{i = 1}^3
(A_{P_i}||A_{C_i})\cong\\ \xymatrix@R=0.5cm{
             \bullet \ar[r]^{a}& \bullet \ar[r]^{e_2}&\bullet \ar[r]^{b}&\bullet \ar[r]^{e_3}&\bullet\\
             \bullet \ar[u]^{e_1} \ar[r]^{a} &\bullet \ar[d]^{e_1} \ar[r]^{e_2}&\bullet \ar[d]^{e_1} \ar[r]^{b}&\bullet \ar[d]^{e_1} \ar[r]^{e_3}&\bullet \ar[d]^{e_1}\\
             \ar[u]& \bullet \ar[r]^{e_2}&\bullet \ar[r]^{b}&\bullet \ar[r]^{e_3}&\bullet\\
             }$
%\xymatrix@R=0.5cm{
%             \bullet \ar[dr]^{a}\\
%             \bullet \ar[u]^{e_1}\ar[dr]^{a}& \bullet \ar[r]^{e_2}&\bullet \ar[r]^{b}&\bullet \ar[r]^{e_3}&\bullet\\
%             \ar[u]& \bullet \ar[d]^{e_1} \ar[r]^{e_2}&\bullet \ar[d]^{e_1} \ar[r]^{b}&\bullet \ar[d]^{e_1} \ar[r]^{e_3}&\bullet \ar[d]^{e_1}\\
%             & \bullet \ar[r]^{e_2}&\bullet \ar[r]^{b}&\bullet \ar[r]^{e_3}&\bullet\\
%             }$
that is bisimilar to $A_S$, i.e., global specification is satisfied
by designing local controllers $A_{C_i}$ to satisfy local
satisfactions $P_i(A_S)$.

Now, suppose that $a$ fails in $E_1$. Since $a$ is passive in $E_1$
and $A_S$ satisfies $EF1$-$EF4$ (since $\delta(q_0,
e_1ae_2be_3)!\wedge \delta(q_0, ae_1e_2be_3)!$ in $A_S$, and hence
$EF1$ and $EF2$ are satisfied; $\Sigma_1 = \{e_1\}$, $\Sigma_2 =
\{a, b, e_2\}$, $\Sigma_3 = \{b, e_3\}$ with the only shared events
$b\in \Sigma_2\cap \Sigma_3$, and the corresponding interleaving
between $F(P_2(A_S)) \cong P_2(A_S)$ and $F(P_3(A_S)) \cong
P_3(A_S)$ is $ae_2be_3$ that appears in $A_S$, with all permutations
with $e_1$ from $F(P_1(A_S))$, and hence, $EF3$ is satisfied, and
finally, $EF4$ is fulfilled since $F(P_1(A_S))$, $F(P_2(A_S))$ and
$F(P_3(A_S))$ are respectively bisimilar to automata
$\xymatrix@R=0.1cm{
                \ar[r]&\bullet \ar[r]^{e_1}&\bullet}$, $\xymatrix@R=0.1cm{
                \ar[r]&\bullet \ar[r]^{a}&\bullet
                \ar[r]^{e_2}&\bullet \ar[r]^{b}&\bullet}$ and $\xymatrix@R=0.1cm{
                \ar[r]&\bullet \ar[r]^{b}&\bullet
                \ar[r]^{e_3}&\bullet}$ that all are deterministic.
Therefore, according to Theorem \ref{Decomposability under event
failure-Theorem}, $\mathop {||}\limits_{i = 1}^3 F(P_i(A_S))\cong
A_S$.

Moreover, since the failed event $a$ is passive in $E_1$ and $A_S$
satisfies $EF1$-$EF4$, as Theorem \ref{Top-Down Approach Under
Failure-n-agents}, the global specification remains satisfied after
failure, as $\mathop {||}\limits_{i = 1}^3 F(A_{P_i}||A_{C_i})\cong
\mathop {||}\limits_{i = 1}^3 F(P_i(A_S))\cong \\ \xymatrix@R=0.5cm{
                \ar[r]&\bullet \ar[d]^{e_1}\ar[r]^{a}&\bullet \ar[d]^{e_1}
                \ar[r]^{e_2}&\bullet \ar[d]^{e_1} \ar[r]^{b}&\bullet \ar[d]^{e_1} \ar[r]^{e_3}&\bullet \ar[d]^{e_1}\\
                 &\bullet \ar[r]^{a}&\bullet
                \ar[r]^{e_2}&\bullet \ar[r]^{b}&\bullet \ar[r]^{e_3}&\bullet}$
that is bisimilar to $A_S$.

\end{example}

\section{Special case: more insight into $2$-agent case}\label{Special case - 2-agent}

This part provides a closer look into the two agent case and
illustrated the notion of global decision making after the event
failures.

First, following lemma presents some properties on a $2$-agent
system that experiences passive failures. The properties will be
then used to provide a deeper insight on the global decision making
of the team on successive and adjacent transitions, in spite of
passive failures.
\begin{lemma}\label{Passivity properties-2-agents} Consider a
deterministic task automaton $A_S  =( Q, q_0 , E = E_1\cup E_2,
\delta)$ and assume that $A_S$ is decomposable with respect to
parallel composition and natural projections $P_i$, $i=1, 2$, and
furthermore assume that $\bar{E}_i = \{a_{i,r}\}$,
$r\in\{1,...,n_i\}$ fail in $E_i$, $i\in\{1, 2\}$. If $\bar{E}_i$,
$i\in\{1, 2\}$ are passive, then $\forall i\in\{1, 2\}$
\begin{enumerate}
\item $\bar{E}_1\cap \bar{E}_2 = \emptyset$;
\item $\bar{E}_1, \bar{E}_2\in E_1\cap E_2$;
\item $\Sigma_1 \backslash \Sigma_2 =  (E_1\backslash
E_2)\cup \bar{E}_2$ and $\Sigma_2 \backslash \Sigma_1 =
(E_2\backslash E_1)\cup \bar{E}_1$.
\end{enumerate}
\end{lemma}

Now, following lemma represents the conditions for maintaining the
capability of a team of two cooperative agents for global decision
making on the orders and selections of transitions in the global
task automaton, after passive event failures.

\begin{lemma}\label{EF1 and EF2 respect to DC1 and DC2}
Consider a deterministic task automaton $A_S = (Q, q_0, E = E_1\cup
E_2, \delta)$. Assume that $A_S$ is decomposable, i.e., $A_S \cong
P_1 (A_S )||P_2 (A_S )$, and furthermore, assume that $\bar{E}_i =
\{a_{i,r}\}$ fail in $E_i$, $r\in\{1,...,n_i\}$, and $\bar{E}_i$ are
passive for $i\in\{1,2\}$. Then, the following two expressions are
equivalent:
\begin{itemize}
\item ($EF1$ and $EF2$): $\forall (e_1, e_2)\in \{(E_1\backslash E_2, \bar{E}_1), ( E_2\backslash
E_1, \bar{E}_2), ( \bar{E}_1, \bar{E}_2)\}$, $q\in Q$, $s\in E^*$:
\begin{eqnarray} [\delta(q, e_1)!\wedge \delta(q, e_2)!]
\Rightarrow [\delta(q,
e_1e_2)!\wedge \delta(q, e_2e_1)!]\label{Switch}\\
\delta(q, e_1e_2s)! \Leftrightarrow \delta(q, e_2e_1s)!\label{Order}
\end{eqnarray}
\item ($DC1_{\Sigma}$ and $DC2_{\Sigma}$): $\forall e_1 \in \Sigma_1\backslash \Sigma_2, e_2 \in \Sigma_2\backslash \Sigma_1, q\in
Q$,  $s\in E^*$: \begin{eqnarray} [\delta(q, e_1)!\wedge \delta(q,
e_2)!] \Rightarrow [\delta(q,
e_1e_2)!\wedge \delta(q, e_2e_1)!]\nonumber \\
\delta(q, e_1e_2s)! \Leftrightarrow \delta(q, e_2e_1s)!.\nonumber
\end{eqnarray}
\end{itemize}
\end{lemma}
\begin{proof}
See the proof in the Appendix.
\end{proof}

\begin{remark}
$EF1$ and $EF2$ represent the decomposability conditions $DC1$ and
$DC2$ after failure, i.e., for the refined local event sets
$\Sigma_1$ and $\Sigma_2$. They say that after the failure, any
decision on the switch or the order between two events that cannot
be accomplished by at least one of the agents ( neither $\{e_1,
e_2\}\subseteq \Sigma_1$, nor $\{e_1, e_2\}\subseteq \Sigma_2$),
then the decision should not be important (both orders should be
legal). This is a good insight on validity of $DC1$ and $DC2$ after
failure of passive events as it is illustrated in Figure
$\ref{Illustration of EF1 and EF2}$, based on the properties in
Lemma \ref{Passivity properties-2-agents}.

From Lemma \ref{Passivity properties-2-agents}, $(\Sigma_1\backslash
\Sigma_2)\times (\Sigma_2\backslash \Sigma_1)$ is the union of four
spaces: $(E_1\backslash E_2)\times (E_2\backslash E_1)$;
$(E_1\backslash E_2)\times (\bar{E}_1)$; $(\bar{E}_2) \times
(E_2\backslash E_1)$, and $(\bar{E}_1) \times (\bar{E}_2)$ (see
Figure $\ref{Illustration of EF1 and EF2} (a) - (d)$). Note that due
to Lemma \ref{Passivity properties-2-agents}, $\bar{E}_1\cap
\bar{E}_2 = \emptyset$.

Now, according to Lemma $\ref{Task Automaton Decomposition}$ in the
Appendix, for any pair of events from $(E_1\backslash E_2)\times
(E_2\backslash E_1)$ (shown in Figure $\ref{Illustration of EF1 and
EF2}- (a)$), $(\ref{Switch})$ and $(\ref{Order})$ are true as $A_S$
is decomposable, before the failure. Moreover, $(\ref{Switch})$ and
$(\ref{Order})$ are also true for the pair of events from other
three spaces of $(\Sigma_1\backslash \Sigma_2)\times
(\Sigma_2\backslash \Sigma_1)$, due to $EF1$ and $EF2$ as it is
illustrated as follows.

\begin{figure}[ihtp]
      \begin{center}
       \includegraphics[width=0.7\textwidth]{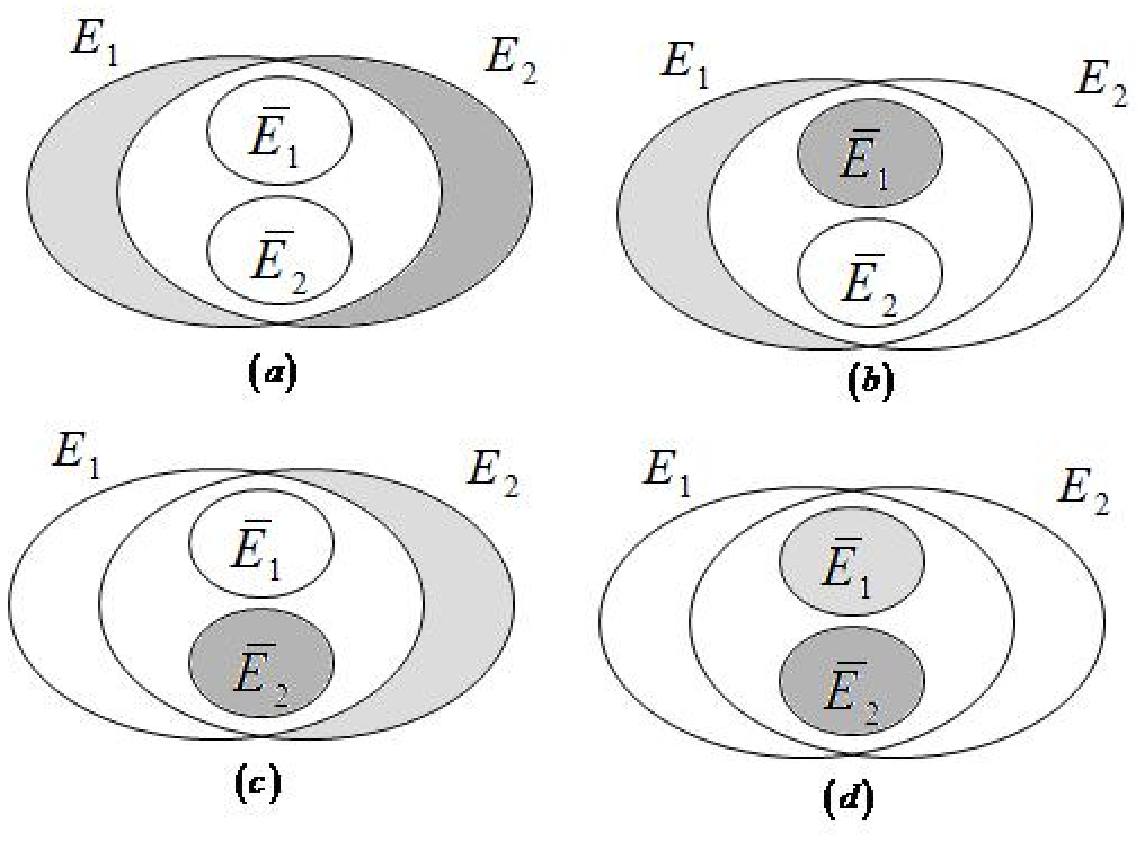}
        \caption{Illustration of $(\Sigma_1\backslash \Sigma_2)\times (\Sigma_2\backslash \Sigma_1) =
[(E_1\backslash E_2)\times (E_2\backslash E_1)] \cup [(E_1\backslash
E_2)\times (\bar{E}_1)]\cup [(\bar{E}_2) \times (E_2\backslash
E_1)]\cup [(\bar{E}_1) \times (\bar{E}_2)]$, (a): $(E_1\backslash
E_2)\times (E_2\backslash E_1)$; (b): $(E_1\backslash E_2)\times
(\bar{E}_1)$; (c): $(\bar{E}_2) \times (E_2\backslash E_1)$, and
(d): $((\bar{E}_1) \times (\bar{E}_2))$.}
         \label{Illustration of EF1 and EF2}
        \end{center}
      \end{figure}

\begin{itemize}
\item Figure $\ref{Illustration of EF1 and
EF2}- (b)$ shows $(E_1\backslash E_2)\times (\bar{E}_1)$: any pair
of events from this space contains in $E_1$, before the failure,
but, contains in neither of $E_1$ and $E_2$ after the failure;
\item Figure $\ref{Illustration of EF1 and
EF2}- (c)$ depicts $(\bar{E}_2) \times (E_2\backslash E_1)$: any
pair of events from this space contains in $E_2$, before the
failure, but, belongs to neither of $E_1$ and $E_2$ after the
failure;
\item Figure $\ref{Illustration of EF1 and
EF2}- (d)$ illustrates $(\bar{E}_1) \times (\bar{E}_2)$: any pair of
events from this space contains in both $E_1$ and $E_2$, before the
failure, but, contains in none of them after the failure.
\end{itemize}
Therefore, since after the failure, for any pair of events from
these three spaces, no agent can be responsible for decision making
on switch/order between them (no local event set contains both
events), then such decisions should not be important as it stated in
$EF1$ and $EF2$.
\end{remark}

Another implication of this result is that when the system is
comprised of only two agents and one of those agent is failed, while
all of its events are passive, then the task automaton remains
decomposable as
\begin{corollary}\label{agent failure for two agents}
Consider a deterministic task automaton $A_S = (Q, q_0, E = E_1\cup
E_2, \delta)$. Assume that $A_S$ is decomposable, i.e., $A_S \cong
P_1 (A_S )||P_1 (A_S )$. Assume furthermore that $E_1$ entirely
fails, i.e., $\bar{E}_1 = E_1$. Then, $A_S \cong \mathop
{||}\limits_{i = 1}^2 F(P_i \left( {A_S } \right))$ if and only if
$\bar{E}_1$ is passive.
\end{corollary}

\begin{proof}
\textbf{Sufficiency:}
 Since $\bar{E}_1 = E_1$, from definition of passivity,
Lemma \ref{Passivity properties-2-agents} and $E = E_1\cup E_2$, it
follows that $E_1 \subseteq E_2 = E$ and $E_1\backslash E_2 =
\bar{E}_2 = \emptyset$, and hence, $EF1$ and $EF2$ hold true, due to
Lemma \ref{EF1 and EF2 respect to DC1 and DC2}. Moreover, since
$\Sigma_1 = E_1\backslash \bar{E}_1 = \emptyset$, then $\Sigma_1
\backslash \Sigma_2 = \Sigma_1 = \emptyset$, that makes $EF3$ always
true. Finally, by Lemma \ref{Passivity properties}, $F(P_1(A_S))$
with $\Sigma_1 = \emptyset$ merges into its initial state, with no
nondeterminism, and $F(P_2(A_S))$ with $\Sigma_2 = E$ is bisimilar
to $A_S$ which is deterministic, therefore, $EF4$ is satisfies, as
well. This implies that when $\bar{E}_1 = E_1$, the passivity of
$\bar{E}_1$ leads to $A_S \cong F(P_1(A_S))||F(P_2(A_S))$.

\textbf{Necessity:} The necessity is proven by contradiction.
Suppose that $\bar{E}_1 = E_1$ and $A_S \cong
F(P_1(A_S))||F(P_2(A_S))$, but $\exists e\in \bar{E}_1$, $e$ is not
passive in $E_1$. Then, from Lemma \ref{Passivity properties}, it is
follows that transitions on $e$ cannot evolve in
$F(P_1(A_S))||F(P_2(A_S))$, due to synchronization constraint in
parallel composition, and hence, $A_S \ncong
F(P_1(A_S))||F(P_2(A_S))$ which is a contradiction.
\end{proof}

\section{Conclusions}\label{Conclusion}
This paper proposed a formal method to investigate whether a
decentralized bisimilarity control design remains valid, under
failure of some events in multi-agent systems.
%Necessary and
%sufficient conditions were proposed for the task automaton to remain
%decomposable, in spite of failures in some shared events (on
%communication links). Furthermore, it was shown that under these
%conditions the global specification remains satisfied, against
%failures on passive events.
%
%Firstly, the passivity of failed events was shown to be a necessary
%condition to preserve the decomposability. Then, we have shown that
%for a deterministic task automaton that experiences failures on
%passive events, the task automaton remains decomposable if and only
%if the important decisions on order/switch between two events can be
%accomplished by at least of one the agents after failure; no illegal
%string is allowed by the composition of local task automata, after
%the failure, and no irremovable new nondeterminism is imposed due to
%the failures.
%
%This work, focuses on faults over communication links and considers
%the fault tolerance from a bisimilarity control point of view, and
%secondly addresses the decentralized structure whose both plan and
%specification are given as composition of local components.
%Particularly,
This work is a continuation of
\cite{Automatica2010-2-agents-decomposability,
TAC2011-n-agents-decomposability}, in which necessary and sufficient
condition was given for task automaton decomposition and the
satisfaction of global specification was guaranteed up on
satisfaction of local specifications. This work then defines
 a new notion of passivity under which it is possible to transform
 the decentralized cooperative control problem under event failures
 into the standard decomposability problem in
 \cite{Automatica2010-2-agents-decomposability, TAC2011-n-agents-decomposability}
 and identifies necessary and sufficient conditions to still guarantee
 the supervised concurrent plant to satisfy the global
 specification, in spite of event failures. The passivity of the
 failed events
 is turned to be a necessary condition for the task automaton to
 remain decomposable, and it is found to reflect the failure of
 redundant communication links. It is then proven that a
 decomposable task automaton remains decomposable and satisfied
 after some passive failures if and only if after the failures, the
 team of agents maintain the capability on collective decision
 making on the orders and selections of transitions and preserve
 the collective perceiving of the task such that the parallel composition of
 local task automata neither allow an illegal behavior (a string
 that is not in the global task automaton), nor disallow a legal
 behavior ( a string from the global task automaton).

This result is of practical importance as it provides a sense of
fault-tolerance to the task decomposition and top-down cooperative
control of multi-agent systems, under event failures.

%A future work in this topic is to make an undecomposable task
%automaton decomposable, by modifying the event distribution among
%the agents.

\section{APPENDIX}
\subsection{Examples for Remark \ref{Remarks for determinism}}
\begin{example}\label{Natural Projection Example}
Following example shows an automaton that does not simulate its
natural projections, yet is decomposable. Consider an automaton
$A_S$: \xymatrix@R=0.1cm{
                \ar[r]&\bullet \ar[r]^{a}\ar[dr]_{e_2}&\bullet   \ar[r]^{e_1}&\bullet \ar[r]^{b}&\bullet\\
             && \bullet \ar[ur]_{e_4} }
with the event set $E=E_1\cup E_2$ and local event sets $E_1=\{a, b,
e_1\}$, $E_2=\{a, b, e_2, e_4\}$. In this example, $A_S$ is
decomposable, since it bisimulates the parallel composition of
$P_1(A_S)$: \xymatrix@C=0.5cm{
     \bullet &  \check{\bullet} \ar@/_/[r]_{a}\ar[l]_{b} &  \bullet \ar[l]_{e_1}
         }
and $P_2(A_S)$:
 \xymatrix@C=0.5cm{
     \ar[r]&  \bullet \ar[r]^{e_2} \ar@/_/[rr]_{a}&  \bullet \ar[r]^{e_4}&
     \bullet \ar[r]^{b}&
     \bullet}, although $P_1(A_S)\not\prec A_S$
     (since the string $b$ appears in $P_1(A_S)$, but not in $A_S$),
     and $P_2(A_S)\not\prec A_S$
     (since the string $ab$ appears in $P_2(A_S)$, but not in $A_S$).
\end{example}

As mentioned in Remark \ref{Remarks for determinism}, another
emergent property is that natural projection of local task automata
may lead to nondeterminism of $P_i \left( A_S \right)$, leading to
nondeterminism of $\overset{n}{\underset{i=1}{\parallel} } P_i
\left( A_S \right)$. The decomposability of $A_S$ again concerns
with bisimilarity of $A_S$ and
$\overset{n}{\underset{i=1}{\parallel} } P_i \left( A_S \right)$,
that may happen even if there exist some nondeterministic
$P_i(A_S)$, as it is elaborated in the following example.
\begin{example}\label{decomposable automaton with
nondeterministic projection} Consider the automaton $A_S$:
\xymatrix@R=0.1cm{
                \ar[r]&  \bullet \ar[r]^{e_1} \ar[dr]_{a}&\bullet \ar[r]^{a} &\bullet \ar[r]^{e_2} &\bullet \\
             && \bullet \ar[r]^{e_2} &\bullet }\\ with $E = E_1\cup E_2$, $E_1 = \{a, e_1\}$, $E_2 = \{a, e_2\}$.
$A_S$ is decomposable, as the parallel composition of
$P_1(A_S)$:\xymatrix@R=0.1cm{
                \ar[r]&  \bullet \ar[r]^{e_1} \ar[dr]_{a}&\bullet \ar[r]^{a} &\bullet  \\
             && \bullet } and $P_2(A_S)$:\\ \xymatrix@R=0.1cm{
                \ar[r]&  \bullet  \ar[dr]_{a} \ar[r]^{a} &\bullet \ar[r]^{e_2} &\bullet \\
             && \bullet \ar[r]^{e_2} &\bullet } is bisimilar to $A_S$. Here, $P_2(A_S)$ is not
                deterministic, but it bisimulates the deterministic
                automaton $P_2(A_S)^{\prime}$: \xymatrix@R=0.1cm{
                \ar[r]& \bullet  \ar[r]^{a}  & \bullet \ar[r]^{e_2} &
                \bullet}.
\end{example}

Therefore, a deterministic task automaton $A_S$ may have
nondeterministic natural projections, and consequently, its
$\overset{n}{\underset{i=1}{\parallel} } P_i \left( A_S \right)$ may
become nondeterministic. As a result, determinism of $A_S$ does not
reduce its decomposability in the sense of bisimulation into its
decomposability in the sense of language equivalence (synthesis
modulo language equivalence \cite{Mukund2002}), due to possibility
of nondeterminism of $\overset{n}{\underset{i=1}{\parallel} } P_i
\left( A_S \right)$, as it is further illustrated in the following
example.

\begin{example}\label{Undecomposable DC3-truncation}
Consider the task automaton $A_S$:\\ \xymatrix@R=0.1cm{
              \ar[r]&  \bullet \ar[r]^{e_1} \ar[dr]_{a}   & \bullet  \ar[r]^{a}  & \bullet\ar[r]^{b}& \bullet  \\
             &    &\bullet
                } \ with $E_1=\{a, b, e_1\}$, $E_2=\{a, b\}$,
               leading to
                $P_1(A_S)$:\xymatrix@R=0.1cm{
          \ar[r]&  \bullet \ar[r]^{e_1} \ar[dr]_{a}  & \bullet  \ar[r]^{a}  & \bullet\ar[r]^{b}& \bullet  \\
             &    &\bullet
                },
                 $P_2(A_S)$:\xymatrix@R=0.1cm{
\ar[r]&  \bullet \ar[r]^{a} \ar[dr]_{a}   & \bullet  \ar[r]^{b}  &  \bullet  \\
             &    & \bullet
                }, and\\
$P_1(A_S)||P_2(A_S)$: \xymatrix@R=0.1cm{
\ar[r]&  \bullet \ar[r]^{e_1} \ar[dr]_{a}   & \bullet  \ar[r]^{a} \ar[dr]_{a} & \bullet\ar[r]^{b}& \bullet  \\
 && \bullet &\bullet
                }\  which is not bisimilar to $A_S$. In this example
 $A_S$ is deterministic, $L(\overset{n}{\underset{i=1}{\parallel} } P_i
\left( A_S \right)) = L(A_S)$; however,
$\overset{n}{\underset{i=1}{\parallel} } P_i \left( A_S \right)
\ncong A_S$.
\end{example}
This example also shows that determinism of $A_S$ also does not
reduce its decomposability in the sense of bisimulation into the
separability of its language (\cite{Willner1991}), as
$\overset{n}{\underset{i=1}{\parallel} } P_i \left( A_S \right)
\ncong A_S$, although $A_S$ is deterministic and its language is
separable ($L(A_S)
  = \overset{n}{\underset{i=1}{|} } L(P_i \left( A_S \right))$).

Therefore, in general for a deterministic task automaton
$\overset{n}{\underset{i=1}{\parallel} } P_i \left( A_S \right)
\cong A_S$ is not reduced into $L(A_S)
 = \overset{n}{\underset{i=1}{|} } L(P_i \left( A_S \right))$. But,
under the determinism of bisimulation quotient of all local task
automata ($DC4$), bisimulation-based decomposability is reduced to
language-based decomposability and the top-down design based on
bisimulation, is reduced to language-based top-down design, such
that the entire closed loop system (the parallel composition of
local closed loop systems) bisimulates (or equivalently is language
equivalent to) the global task automaton. In case of $DC4$, the
other three conditions ($DC1$-$DC3$) can be used to characterize the
language separability.

\subsection{Proof for Lemma \ref{Passivity properties}}
Firstly, in order to allow the global transitions, the failed event
$a$ in $E_i$ has to be received from other agents not from its own
sensors and actuator readings, otherwise, no local transitions on
$a$ evolve in either of $F(A_i)$ or $\overset{n}{\underset{i=1}{||}
}F(A_i)$ (since other agents receive $a$ from $A_i$). Therefore, the
failed events have to necessarily be shared events ($loc(a) >1$),
and that after the failure of $a$ in $A_i$, $a$ is excluded from
$E_i$, i.e., $\Sigma_i = E_i\backslash a$, as $a$ is not received to
$A_i$ from other agents. Moreover, due to Definition \ref{parallel
composition}, exclusion of $a$ from $E_i$ allows global transitions
on $a$ with no synchronization restriction from $F(A_i)$. Finally,
the transitions on failed event $a$ has to be replaced with
$\varepsilon$-moves, in order to allow transitions after $a$ in
$A_i$, i.e., $\forall x_1, x_2 \in Q_i$, $\delta_i(x_1, a) = x_2$,
then $\delta^F_i([x_1]_{\Sigma_i}, a) = [x_2]_{\Sigma_i}$,
$[x_1]_{\Sigma_i} = [x_2]_{\Sigma_i}$ and $F(A_i) = P_{E_i\backslash
a}(A_i)$ (otherwise a transition of $\delta_i^F(\delta_i^F(x, a),
e)$ will be disabled due to stopping of execution of $\delta_i^F(x,
a)$). It should be noted that, if there are no traditions after
$\delta_i(x, a)$ (i.e., $\forall e\in E_i$:
$\neg\delta_i(\delta_i(x, a), e)!$, then stoping of $\delta_i(x, a)$
is identical to replacing this transition with an
$\varepsilon$-move.
%
%Therefore, $\forall z_1, z_2 \in Z$, $\delta_{||}(z_1, a) = z_2$ and
%$\delta_{||}^F(z_1, a) = z_2$, after failure of $a$ in $A_i$,
%requires that $\Sigma_i = E_i\backslash a$, $E =
%\overset{n}{\underset{i=1}{\cup}}\Sigma_i$, and $F(A_i) =
%P_{E_i\backslash a}(A_i)$, that collectively result in $\forall x_1,
%x_2 \in Q_i$, $\delta_i(x_1, a) = x_2$, then
%$\delta^F_i([x_1]_{\Sigma_i}, a) = [x_2]_{\Sigma_i}$,
%$[x_1]_{\Sigma_i} = [x_2]_{\Sigma_i}$ and $\forall j\in\{1, \ldots,
%n\}$, if $a$ does not fail in $A_j$ then $a \in\Sigma_j$, $\forall
%x_1, x_2 \in Q_j$, $\delta_j(x_1, a) = x_2$, then
%$\delta^F_j([x_1]_{\Sigma_j}, a) = [x_2]_{\Sigma_j}$. From
%definition of parallel composition, therefore, $\forall z_1, z_2 \in
%Z$, $\delta_{||}(z_1, a) = z_2$ if $\exists x_1, x_2 \in Q_i$,
%$\delta_i(x_1, a) = x_2$, $z_1[i] = x_1$, $z_2[i] = x_2$, where
%$z_1[i]$ is the $i-th$ component of $z$ in
%$\overset{n}{\underset{i=1}{||}}F(A_i)$.
These collectively mean that preserving of global transitions in
$\overset{n}{\underset{i=1}{||}}F(A_i)$ requires then local failures
to be passive.
%
%For a non-passive failure, on the other hand, the existence of
%$i\in\{1, \ldots, n\}$, such that $\Sigma_i = E_i$, $\delta_i(x_1,
%e) = x_2$, and $\neg\delta^F_i([x_1]_{\Sigma_i}, e)!$ results in
%$\neg\delta_{||}(z_1, e)!$ for any $x_1\in Q_i$, $\delta_i(x_1, e)!
%$, $z_1[i] = x_1$.

\subsection{Proof for Lemma \ref{EF1,2 and DC1,2}}
Passivity of all $\bar{E}_i$, $i\in \{1, \ldots, n\}$, due to
definition of passivity, leads to $\Sigma_i = E_i\backslash
\bar{E}_i \subseteq E_i$, and hence, the expression $[\exists
E_i\in\{E_1, \cdots, E_n\}, \{e_1, e_2\}\subseteq E_i]$ in the
antecedent of $DC1$ and $DC2$ leads to $[\exists
\Sigma_i\in\{\Sigma_1, \cdots, \Sigma_n\}, \{e_1, e_2\}\subseteq
\Sigma_i]$, replacing $E_i$ with $\Sigma_i = E_i\backslash
\bar{E}_i$.

\subsection{Proof for Lemma \ref{EF4 and DC4}}
Any nondeterminism in $F(P_i(A_S))$ appears either due to
nondeterminism from $P_i(A_S)$ or newly formed nondeterminism
because of replacing of passive events by $\varepsilon$.

In the first case, from decomposability of $A_S$, $DC4$ says that
for any $x, x_1, x_2 \in Q_i$, $e\in E_i\backslash \bar{E}_i$, $t\in
E_i^*$, $x_1 \neq x_2$, $\delta_i (x, e) = x_1$, $\delta_i (x, e) =
x_2$: $\delta_i (x_1, t)! \Leftrightarrow \delta_i (x_2, t)!$, i.e.,
$\delta_i^F([x]_{\Sigma_i}, e) = [x_1]_{\Sigma_i}$,
$\delta_i^F([x]_{\Sigma_i}, e) = [x_2]_{\Sigma_i}$:
$\delta_i^F([x_1]_{\Sigma_i}, p_{\Sigma_i}(t))! \Leftrightarrow
\delta_i^F([x_2]_{\Sigma_i}, p_{\Sigma_i}(t))!$, which is $DC4$ for
$F(P_i(A_S))$, with refined local event set $\Sigma_i$.

For the second case, any newly appeared nondeterminism is induced by
transitions from the original local task automat, in the following
form. $\exists i\in\{1, \ldots, n\}, x, x_1, x_2 \in Q_i$, $t_1\in
\bar{E}_i^*$, $e\in E_i\backslash \bar{E}_i$, $t\in E_i^*$, $x_1
\neq x_2$, $\delta_i (x, t_1e) = x_1$, $\delta_i (x, e) = x_2$ then
$[x]_{\Sigma_i} = [\delta_i^F ([x]_{\Sigma_i}, t_1)]_{\Sigma_i}$,
and hence, $EF4$ becomes $\delta_i^F ([x]_{\Sigma_i}, e) =
[x_1]_{\Sigma_i}$, $\delta_i^F ([x]_{\Sigma_i}, e) =
[x_2]_{\Sigma_i}$: $\delta_i^F ([x_1]_{\Sigma_i}, p_{\Sigma_i}(t))!
\Leftrightarrow \delta_i^F ([x_2]_{\Sigma_i}, p_{\Sigma_i}(t))!$,
which is again equivalent to $DC4$ for $F(P_i(A_S))$.

\subsection{Proof for Lemma \ref{Passivity properties-2-agents}}
The first item is proven based on the fact that if $\exists e\in
\bar{E}_1\cap \bar{E}_2$, then $snd_e(1) = \emptyset \wedge snd_e(2)
= \emptyset$ which is impossible, due to Remark \ref{Meaning of
Passivity} that requires $send_e(i) = \emptyset \wedge rec_e(i) \neq
\emptyset$ for an event $e$ to be passive in $E_i \in \{E_1, E_2\}$,
in two agent case.

The second item, comes from passivity of $\bar{E}_1$ and $\bar{E}_2$
that implies that $\forall e\in \bar{E}_i$, $i = 1, 2$, $snd_e(i) =
\emptyset\wedge rcv_e (i)\neq \emptyset$, and hence $loc(e)>1$ which
means $e\in E_j$, $j\in\{1, 2\}\backslash \{i\}$, i.e., $e\in
E_1\cap E_2$.

For the last item, from the second item and $\bar{E}_1\cap \bar{E}_2
= \emptyset$ we respectively have $\bar{E}_1, \bar{E}_2 \subseteq
E_1\cap E_2$ and $\bar{E}_1\subseteq \bar{E}_2^{\prime}$,
$\bar{E}_2\subseteq \bar{E}_1^{\prime}$ (In this proof, prime
operation stands for the set complements, where the $E_1\cup E_2$ is
considered as the universal set). Consequently, $\Sigma_1 \backslash
\Sigma_2$ = $(E_1\backslash \bar{E}_1)\backslash (E_2\backslash
\bar{E}_2)$ = $(E_1\cap \bar{E}_1^{\prime})\cap \left(E_2\cap
\bar{E}_2^{\prime}\right)^{\prime}$ = $(E_1\cap
\bar{E}_1^{\prime})\cap \left(E_2^{\prime}\cup \bar{E}_2\right)$ =
$[(E_1\cap \bar{E}_1^{\prime})\cap E_2^{\prime}] \cup
 [(E_1\cap \bar{E}_1^{\prime})\cap \bar{E}_2]$ =
$[E_1\cap \left(\bar{E}_1\cup E_2\right)^{\prime}] \cup
 [(E_1\cap \bar{E}_2) \cap \bar{E}_1^{\prime} ]$ =
$(E_1\cap E_2^{\prime}) \cup
 (\bar{E}_2 \cap \bar{E}_1^{\prime})$ =
$(E_1\backslash E_2) \cup
 \bar{E}_2 $. Similarly,
$\Sigma_2 \backslash \Sigma_1 = (E_2\backslash E_1)\cup \bar{E}_1$.

\subsection{Proof for Lemma \ref{EF1 and EF2 respect to DC1 and
DC2}}
%Firstly, the equivalence of $DC1$ and $DC2$ in Lemmas
%\ref{Task Automaton Decomposition} and \ref{EF1 and EF2 respect to
%DC1 and DC2} is given by the following lemma, substituting $n = 2$.
%\begin{lemma}\label{Decision Making n agents}(Lemma $6$ in
%\cite{TAC2011-n-agents-decomposability}) Consider a deterministic
%automaton $A_S = (Q, q_0, E = \mathop{||}\limits_{i = 1}^n\cup E_i,
%\delta)$ and natural projections $P_i$, $i=1,...,n$. Then following
%statements are equivalent
%\begin{enumerate}\item $\forall e_1 ,e_2 \in E,q \in Q, s \in E^* ,\forall E_i \in
%\left\{ {E_1 ,...,E_n } \right\},\left\{ {e_1 ,e_2 } \right\}
%\not\subset E_i $:
%\begin{itemize}
%\item $DC1$: $[\delta (q,e_1)! \wedge \delta (q,e_2)! ]
%  \Rightarrow [\delta (q,e_1 e_2)! \wedge \delta (q,e_2 e_1 )!]$;
% \item $DC2$: $\delta \left( {q,e_1 e_2 s} \right)! \Leftrightarrow \delta \left( {q,e_2 e_1 s}
% \right)!$;
% \end{itemize}
%\item
%\begin{itemize}
%\item $DC1$: $\forall e_1,
%e_2 \in E, q\in Q$: $[\delta(q,e_1)!\wedge \delta(q,e_2)!]\\
%\Rightarrow [\exists E_i\in\{E_1, \ldots, E_n\}, \{e_1,
%e_2\}\subseteq E_i]\vee[\delta(q, e_1e_2)! \wedge \delta(q,
%e_2e_1)!]$;
%\item $DC2$: $\forall e_1, e_2 \in E,  q\in Q$, $s\in E^*$: $[\delta(q,
%e_1e_2s)!\vee \delta(q, e_2e_1s)!]\\ \Rightarrow [\exists
%E_i\in\{E_1, \ldots, E_n\}, \{e_1, e_2\}\subseteq E_i]\vee [
%\delta(q, e_1e_2s)!\wedge \delta(q, e_2e_1s)!]$.
%\end{itemize}
%\end{enumerate}
%\end{lemma}
%Now,
To prove this lemma, firstly, the decomposability result for two
agents is recalled as
\begin{lemma} (Theorem 1 in \cite{Automatica2010-2-agents-decomposability})) \label{Task Automaton
Decomposition} A deterministic  automaton $A_S=(Q, q_0, E=E_1\cup
E_2, \delta)$ is decomposable with respect to parallel composition
and natural projections $P_i$, $i=1,2$, such that $A_S\cong
P_1(A_S)||P_2(A_S)$ if and only if it satisfies the following
decomposability conditions: $\forall e_1 \in E_1\backslash E_2, e_2
\in E_2\backslash E_1, q\in Q$, $s\in E^*$,
\begin{itemize}\item $DC1$: $[\delta(q,e_1)!\wedge
\delta(q,e_2)!]\Rightarrow [\delta(q, e_1e_2)! \wedge \delta(q,
e_2e_1)!]$;
\item $DC2$: $\delta(q, e_1e_2s)!\Leftrightarrow \delta(q,
e_2e_1s)!$;
 \item $DC3$:
$\forall s, s^{\prime} \in E^*$, sharing the same first appearing
common event $a\in E_1 \cap E_2$, $s\neq s^{\prime}$, $q\in Q$:
$\delta(q, s)! \wedge \delta(q, s^{\prime})! \Rightarrow \delta(q,
p_1(s)|p_2(s^{\prime}))! \wedge \delta(q, p_1(s^{\prime})|p_2(s))!$,
and
\item $DC4$: $\forall i\in\{1, 2\}$, $x, x_1, x_2 \in Q_i$, $x_1\neq x_2$, $e\in E_i$, $t\in E_i^*$,
 $\delta_i (x, e)=  x_1$,  $\delta_i (x, e)=
x_2$: $\delta_i (x_1, t)! \Leftrightarrow \delta_i(x_2, t)!$.
\end{itemize}
\end{lemma}

Now, in order to prove the equivalence of two cases in lemma
\ref{EF1 and EF2 respect to DC1 and DC2}, one needs to prove that
the set $\{\bar{E}_1 \times E_1\backslash E_2, \bar{E}_2 \times
E_2\backslash E_1, \bar{E}_1 \times \bar{E}_2\}$ in $EF1$ and $EF2$
is equal to the set $\{(\Sigma_1\backslash \Sigma_2) \times
(\Sigma_2\backslash \Sigma_1)\}$ in $DC1_{\Sigma}$ and
$DC1_{\Sigma}$ (decomposability conditions $DC1$ and $DC2$ with
respect to $\Sigma_1$ and $\Sigma_2$).

From lemma \ref{Passivity properties-2-agents}, $e_1\in \Sigma_1
\backslash \Sigma_2$, $e_2\in \Sigma_2 \backslash \Sigma_1$  is
equivalent to $e_1\in (E_1\backslash E_2)\cup \bar{E}_2$, $e_2 \in
(E_2\backslash E_1)\cup \bar{E}_1$ which means that $e_1\in
E_1\backslash E_2 \vee e_1\in \bar{E}_2$ and $e_2\in E_2\backslash
E_1 \vee e_2\in \bar{E}_1$, leading to four possible cases: $(e_1\in
E_1\backslash E_2\wedge e_2\in E_2\backslash E_1)$, $(e_1\in
E_1\backslash E_2\wedge e_2\in \bar{E}_1)$, $(e_1\in \bar{E}_2\wedge
e_2\in E_2\backslash E_1)$ or $(e_1\in \bar{E}_2\wedge e_2\in
\bar{E}_1)$.

Now, Lemma \ref{EF1 and EF2 respect to DC1 and DC2} is proven as
follows. For the first case, since decomposability of $A_S$ implies
$DC1$ and $DC2$, then, $\forall e_1\in E_1\backslash E_2, e_2\in
E_2\backslash E_1$, $q\in Q$, $s\in E^*$: $(\ref{Switch})$ and
$(\ref{Order})$ hold true. For the second, third and fourth cases,
i.e., when $(e_1\in E_1\backslash E_2\wedge e_2\in \bar{E}_1)$,
$(e_1\in \bar{E}_2\wedge e_2\in E_2\backslash E_1)$ or $(e_1\in
\bar{E}_2\wedge e_2\in \bar{E}_1)$, then $(\ref{Switch})$ and
$(\ref{Order})$ are guarantee by $EF1$ and $EF2$. Therefore,
provided the decomposability of $A_S$, $EF1$ and $EF2$,
$(\ref{Switch})$ and $(\ref{Order})$ become true for all $e_1\in
\Sigma_1\backslash \Sigma_2$, $e_2\in \Sigma_2\backslash \Sigma_1$.
This means that $EF1$ and $EF2$ are respectively equivalent to $DC1$
and $DC2$ after failures (for $\Sigma_1$ and $\Sigma_2$).

\bibliographystyle{IEEEtran}
\bibliography{Ref_9}

\begin{thebibliography}{10}
\providecommand{\url}[1]{#1}
\csname url@rmstyle\endcsname
\providecommand{\newblock}{\relax}
\providecommand{\bibinfo}[2]{#2}
\providecommand\BIBentrySTDinterwordspacing{\spaceskip=0pt\relax}
\providecommand\BIBentryALTinterwordstretchfactor{4}
\providecommand\BIBentryALTinterwordspacing{\spaceskip=\fontdimen2\font plus
\BIBentryALTinterwordstretchfactor\fontdimen3\font minus
  \fontdimen4\font\relax}
\providecommand\BIBforeignlanguage[2]{{%
\expandafter\ifx\csname l@#1\endcsname\relax
\typeout{** WARNING: IEEEtran.bst: No hyphenation pattern has been}%
\typeout{** loaded for the language `#1'. Using the pattern for}%
\typeout{** the default language instead.}%
\else
\language=\csname l@#1\endcsname
\fi
#2}}

\bibitem{Automatica2010-2-agents-decomposability}
\BIBentryALTinterwordspacing
M.~Karimadini and H.~Lin, ``Guaranteed global performance through local
  coordinations,'' \emph{Automatica}, vol. In Press, Corrected Proof, 2011.
  [Online]. Available:
  \url{http://www.sciencedirect.com/science/article/B6V21-529M7SH-5/2/19d8df66%
047d3ddc475af4b1b79ef8d2}
\BIBentrySTDinterwordspacing

\bibitem{TAC2011-n-agents-decomposability}
------, ``Necessary and sufficient conditions for task automaton
  decomposition,'' \emph{submitted to IEEE Transactions on Automatic Control},
  vol. 2011, 2011.

\bibitem{Lima2005}
P.~U. Lima and L.~M. Custódio, \emph{Multi-Robot Systems, Book Series Studies
  in Computational Intelligence, Book Innovations in Robot, Mobility and
  Control}.\hskip 1em plus 0.5em minus 0.4em\relax Berlin: Springer Berlin /
  Heidelberg, 2005, vol.~8.

\bibitem{Georgios2009}
G.~E. Fainekos, A.~Girard, H.~Kress-Gazit, and G.~J. Pappas, ``Temporal logic
  motion planning for dynamic robots,'' \emph{Automatica}, vol.~45, no.~2, pp.
  343--352, 2009.

\bibitem{Ji2009}
Z.~Ji, Z.~Wang, H.~Lin, and Z.~Wang, ``Brief paper: Interconnection topologies
  for multi-agent coordination under leader-follower framework,''
  \emph{Automatica}, vol.~45, no.~12, pp. 2857--2863, 2009.

\bibitem{Lesser1999}
V.~R. Lesser, ``Cooperative multiagent systems: A personal view of the state of
  the art,'' \emph{IEEE Transactions on Knowledge and Data Engineering},
  vol.~11, pp. 133--142, 1999.

\bibitem{Choi2009}
J.~Choi, S.~Oh, and R.~Horowitz, ``Distributed learning and cooperative control
  for multi-agent systems,'' \emph{Automatica}, vol.~45, no.~12, pp.
  2802--2814, 2009.

\bibitem{Kazeroon2009}
E.~Semsar-Kazerooni and K.~Khorasani, ``Multi-agent team cooperation: A game
  theory approach,'' \emph{Automatica}, vol.~45, no.~10, pp. 2205--2213, 2009.

\bibitem{Mukund2002}
M.~Mukund, \emph{From global specifications to distributed implementations, in
  B. Caillaud, P. Darondeau, L. Lavagno (Eds.), \emph{Synthesis and Control of
  Discrete Event Systems}, Kluwer}.\hskip 1em plus 0.5em minus 0.4em\relax
  Berlin: Springer Berlin / Heidelberg, 2002.

\bibitem{Willner1991}
Y.~Willner and M.~Heymann, ``Supervisory control of concurrent discrete-event
  systems,'' \emph{International Journal of Control}, vol.~54, pp. 1143--1169,
  1991.

\bibitem{Sampath1995}
M.~Sampath, R.~Sengupta, S.~Lafortune, K.~Sinnamohideen, and D.~Teneketzis,
  ``Diagnosability of discrete-event systems,'' \emph{Automatic Control, IEEE
  Transactions on}, vol.~40, no.~9, pp. 1555 --1575, Sept. 1995.

\bibitem{Takai2000}
S.~Takai and T.~Ushio, ``Reliable decentralized supervisory control of discrete
  event systems,'' \emph{Systems, Man, and Cybernetics, Part B: Cybernetics,
  IEEE Transactions on}, vol.~30, no.~5, pp. 661 --667, Oct. 2000.

\bibitem{Liu2010}
\BIBentryALTinterwordspacing
F.~Liu and H.~Lin, ``Reliable supervisory control for general architecture of
  decentralized discrete event systems,'' \emph{Automatica}, vol.~46, no.~9,
  pp. 1510 -- 1516, 2010. [Online]. Available:
  \url{http://www.sciencedirect.com/science/article/B6V21-50FHMXJ-3/2/49559880%
c5e773d537069b2f5ac6e318}
\BIBentrySTDinterwordspacing

\bibitem{Lin1993}
F.~Lin, ``Robust and adaptive supervisory control of discrete event systems,''
  \emph{Automatic Control, IEEE Transactions on}, vol.~38, no.~12, pp. 1848
  --1852, Dec. 1993.

\bibitem{Darabi2003}
H.~Darabi, M.~Jafari, and A.~Buczak, ``A control switching theory for
  supervisory control of discrete event systems,'' \emph{Robotics and
  Automation, IEEE Transactions on}, vol.~19, no.~1, pp. 131 -- 137, Feb. 2003.

\bibitem{Rohloff2005}
K.~Rohloff, ``Sensor failure tolerant supervisory control,'' in \emph{Decision
  and Control, 2005 and 2005 European Control Conference. CDC-ECC '05. 44th
  IEEE Conference on}, 2005, pp. 3493 -- 3498.

\bibitem{Lafortune1990}
S.~Lafortune and F.~Lin, ``On tolerable and desirable behaviors in supervisory
  control of discrete event systems,'' in \emph{Decision and Control, 1990.,
  Proceedings of the 29th IEEE Conference on}, Dec. 1990, pp. 3434 --3439
  vol.6.

\bibitem{Jensen2003}
R.~M. Jensen, \emph{A survey of formal methods for intelligent swarmsDES
  Controller Synthesis and Fault Tolerant Control: A Survey of Recent Advances,
  Tech. Rep. TR-2003-40}.\hskip 1em plus 0.5em minus 0.4em\relax IT Univ. of
  Copenhagen, Copenhagen, Denmark: NASA Goddard Space Flight Center, 2003.

\bibitem{Qin2008}
Q.~Wen, R.~Kumar, J.~Huang, and H.~Liu, ``A framework for fault-tolerant
  control of discrete event systems,'' \emph{Automatic Control, IEEE
  Transactions on}, vol.~53, no.~8, pp. 1839 --1849, 2008.

\bibitem{Brave1990}
Y.~Brave and M.~Heymann, ``On stabilization of discrete event processes,''
  \emph{Int. J. Control}, vol.~51, no.~5, pp. 1101–--1117, 1990.

\bibitem{Ozveren1991}
C.~M. \"{O}zveren, A.~S. Willsky, and P.~J. Antsaklis, ``Stability and
  stabilizability of discrete event dynamic systems,'' \emph{J. ACM}, vol.~38,
  no.~3, pp. 729--751, 1991.

\bibitem{Kumar1993}
R.~Kumar, V.~Garg, Marcus, and S.~I. Marcus, ``Language stability and
  stabilizability of discrete event dynamical systems,'' \emph{SIAM Journal of
  Control and Optimization}, vol.~31, no.~5, pp. 1294–--1320, 1993.

\bibitem{Willner1995}
Y.~Willner and M.~Heymann, ``Language convergence in controlled discrete-event
  systems,'' \emph{Automatic Control, IEEE Transactions on}, vol.~40, no.~4,
  pp. 616 --627, Apr. 1995.

\bibitem{Saboori2005}
A.~Saboori and S.~Zad, ``Fault recovery in discrete event systems,'' in
  \emph{Computational Intelligence Methods and Applications, 2005 ICSC Congress
  on}, 0 2005.

\bibitem{CCC2011-n-agents-Event-Failure}
M.~Karimadini and H.~Lin, ``Reliable task decomposability for cooperative
  multi-agent systems,'' in \emph{submitted to the 30th Chinese Control
  Conference, CCC2011}, 2011.

\bibitem{Kumar1999}
R.~Kumar and V.~K. Garg, \emph{Modeling and Control of Logical Discrete Event
  Systems}.\hskip 1em plus 0.5em minus 0.4em\relax Norwell, MA, USA: Kluwer
  Academic Publishers, 1999.

\bibitem{Cassandras2008}
C.~G. Cassandras and S.~Lafortune, \emph{Introduction to discrete event
  systems}.\hskip 1em plus 0.5em minus 0.4em\relax USA: Springer, 2008.

\bibitem{Zhou2006}
C.~Zhou, R.~Kumar, and S.~Jiang, ``Control of nondeterministic discrete-event
  systems for bisimulation equivalence,'' \emph{Automatic Control, IEEE
  Transactions on}, vol.~51, no.~5, pp. 754 -- 765, may 2006.

\bibitem{Alur2000}
R.~Alur, T.~Henzinger, G.~Lafferriere, and G.~Pappas, ``Discrete abstractions
  of hybrid systems,'' \emph{Proceedings of the IEEE}, vol.~88, no.~7, pp. 971
  --984, jul 2000.

\bibitem{Morin1998}
R.~Morin, ``Decompositions of asynchronous systems,'' in \emph{CONCUR '98:
  Proceedings of the 9th International Conference on Concurrency Theory}.\hskip
  1em plus 0.5em minus 0.4em\relax London, UK: Springer-Verlag, 1998, pp.
  549--564.

\end{thebibliography}
%% not want to use model5-names.bst.



%% \bibitem must have one of the following forms:
%%   \bibitem[Jones et al.(1990)]{key}...
%%   \bibitem[Jones et al.(1990)Jones, Baker, and Williams]{key}...
%%   \bibitem[Jones et al., 1990]{key}...
%%   \bibitem[\protect\citeauthoryear{Jones, Baker, and Williams}{Jones
%%       et al.}{1990}]{key}...
%%   \bibitem[\protect\citeauthoryear{Jones et al.}{1990}]{key}...
%%   \bibitem[\protect\astroncite{Jones et al.}{1990}]{key}...
%%   \bibitem[\protect\citename{Jones et al., }1990]{key}...
%%   \harvarditem[Jones et al.]{Jones, Baker, and Williams}{1990}{key}...
%%


% \end{thebibliography}

\end{document}